\documentclass[com,crcready]{iosart2x}

\usepackage{amsmath}
\usepackage{amsfonts}
\IfFileExists{pdfsync.sty}{\usepackage{pdfsync}}{}
\usepackage{hyperref}

\usepackage{amssymb}
\providecommand\LANGUE[2]{#1}
\renewcommand{\restriction}{\mathord{\upharpoonright}}
\providecommand\spnewtheorem[4]{\newtheorem{#1}{#2}}
\spnewtheorem{definition}{Definition}{\bfseries}{\rmfamily}
\spnewtheorem{theorem}{\LANGUE{Theorem}{Théorème}}{\bfseries}{\rmfamily}
\spnewtheorem{remark}{\LANGUE{Remark}{Remarque}}{\bfseries}{\rmfamily}
\spnewtheorem{corollary}{\LANGUE{Corollary}{Corollaire}}{\bfseries}{\rmfamily}
\spnewtheorem{example}{\LANGUE{Example}{Exemple}}{\bfseries}{\rmfamily}
\spnewtheorem{proposition}{Proposition}{\bfseries}{\rmfamily}

\usepackage[utf8]{inputenc}
\usepackage[T1]{fontenc}

\begin{document}

\begin{frontmatter}

\newcommand\myaddress{ 
LIX,  Ecole Polytechnique, Institut Polytechnique de Paris, France 
}

\newcommand\hisaddress{University Paris-Est Créteil Val de Marne, LACL 
}

\title{Set Descriptive Complexity of Solvable Functions} 
\runtitle{Set Descriptive Complexity of Solvable Functions}

\begin{aug}
\author[A]{\inits{}\fnms{Olivier} \snm{Bournez}\ead[label=e1]{olivier.bournez@lix.polytechnique.fr}}
\author[B]{\inits{}\fnms{Riccardo} \snm{Gozzi}\ead[label=e2]{ilgozzi@mail.com}}
\address[A]{LIX, \orgname{Ecole Polytechnique, Institut Polytechnique de Paris} 
 \cny{France}\printead[presep={\\}]{e1}}
\address[B]{LACL,  \orgname{University Paris-Est Créteil Val de Marne,},
 \cny{France}\printead[presep={\\}]{e2}}
\address[C]{Department first, \orgname{University or Company name},
Abbreviate US states, \cny{Country}\printead[presep={\\}]{e3}}
\end{aug}

%

%

\maketitle

\newcommand{\interior}[1]{%
  {\kern0pt#1}^{\mathrm{o}}%
}
\newcommand\PRIMOF[1]{\tu f_{[#1]}}

\begin{abstract}
In a recent article, we introduced and studied a precise class of dynamical systems called \emph{solvable systems}. These systems present a dynamic ruled by discontinuous ordinary differential equations with \emph{solvable} right-hand terms and unique evolution. They correspond to a class of systems for which a transfinite method exist to compute the solution. We also presented several examples including a nontrivial one whose solution yields, at an integer time, a real encoding of the halting set for Turing machines; therefore showcasing that the behavior of solvable systems might describe ordinal Turing computations. 

In the current article, we study in more depth solvable systems, using tools from descriptive set theory. 
By establishing a correspondence with the class of well-founded trees, we construct a coanalytic ranking over the set of solvable functions and discuss its relation with other existing rankings for differentiable functions, in particular with the Kechris-Woodin, Denjoy and Zalcwasser ranking.  We prove that our ranking is unbounded below the first uncountable ordinal.

\end{abstract}

\end{frontmatter}

\section{Introduction} 

The scope of this paper is to explore the definability properties of solutions of dynamical systems defined by (possibly discontinuous) Ordinary Differential Equations (ODEs), as well as investigate their potential connection with models of computations. 

\paragraph{Polynomial ODEs and models of computations} 
In a broad sense, this investigation is related to analog models of computation. It  started with Claude Shannon, when he first introduced his \emph{General Purpose Analog Model of computation} (GPAC) \cite{Sha41} to provide a proper theoretical framework for characterizing functions computable by integrator machines called differential analyzers. The analog model of computation that was founded this way, featuring continuous-time and space, was initially suspected to include all differential algebraic functions. It was later proved in \cite{GC03} that a polished version of such a model described precisely all functions that could be obtained as solutions of Initial Value Problems (IVPs) involving polynomial ODEs. 

The importance of polynomial ODEs for modeling real-world dynamics is well known across multiple scientific disciplines, especially physics, so their pivoting role within analog computation showcased in \cite{GC03} must sound logical to the expert ear. Nonetheless, the correspondence between polynomial ODEs and discrete computation demonstrated in \cite{TAMC06} came as a greater surprise. Indeed, the authors of \cite{TAMC06} successfully proved the existence of an interesting equivalence between the GPAC model and computable analysis by constructively designing analog simulations of Turing machines computations capable of generating every computable function over the reals. This result prompted the development of a rich research endeavor that focused on identifying, within the context of this established equivalence, the right analog descriptions for known complexity classes such as FP \cite{JournalACM2017}, FEXPTIME \cite{GozziGraca2022} or FPSPACE \cite{BGDPRiccardo2022}. 

The philosophy of all the above investigations was to extend the meaning of the equivalence in \cite{TAMC06} by subdividing it into precise subclasses concerning complexity stratifications. 

\paragraph{Computing with more general ODEs}
In a recent article, \cite{StacsBournezGozzi2024} we chose a radically different approach by looking at an alternative implication of the result. Instead of inspecting the consequences that could spread \emph{from} the equivalence, we decided to analyze the assumptions that lead \emph{to} the equivalence. One key requirement leading to \cite{TAMC06} was the possibility of only making use of polynomial ODEs to build the continuous-time simulations of the Turing machines considered. This requirement followed directly from the necessity of making use of GPAC generable functions in the sense of \cite{GC03}. Instead, one of the goals of \cite{StacsBournezGozzi2024} was to understand the repercussions of broadening the spectrum of the IVPs and ODEs considered while still maintaining the spirit of \cite{TAMC06} focused on continuous-time simulations of discrete dynamics. In other words, one purpose of \cite{StacsBournezGozzi2024}  was to set the premises to address the following question: is it possible to design continuous-time dynamical systems ruled by (hence non-polynomial) ODEs that are capable of simulating ordinal Turing computations in a precise mathematical sense? 

We decided to tackle this question from a specific angle: first, we unequivocally defined what we meant by ordinal Turing computations by making use of concepts from set descriptive theory and ordinal computing. Second, we identified which class of ODEs should be used to obtain simulations with this desired behavior, and we presented strong reasons to support a positive answer to the question raised above. 

More precisely, we called this class the class of \emph{solvable} ODEs and we proved that IVPs involving such dynamics can always be solved analytically via transfinite recursion bounded by the first uncountable ordinal. 

\paragraph{About solvable ODEs}
Solvable IVPs as we defined them allow for the presence of discontinuities in their right-hand terms, enabling their solutions to acquire transfinite properties, while at the same featuring enough regularity to be successfully singled out within the domain using only the common tools from analysis. One of the key characteristics of solvable IVPs is that they exhibit unique, everywhere differentiable solutions by hypothesis, making them stand out from the most common approaches to the theory of discontinuous ODEs such as the ones listed in \cite{FilippovBook,aubin2012differential,deimling2011multivalued}. This feature naturally promotes them as perfect candidates for the type of characterization that we are seeking. We produced another promising step in this direction with one of the examples discussed in \cite{StacsBournezGozzi2024}, where we presented the details of one solvable system whose solution assumes, at an integer time, a real value encoding the halting set for Turing machines, therefore placing the first milestone in the direction of achieving higher Turing jumps in the hyperarithmetical hierarchy.

\paragraph{Relations to other rankings and hiearchies} 
This type of achievement has to be intended according to the methodology explained by Dougherty and Kechris for the context of Denjoy integration, who showed that each hyperarithmetical real can be obtained through Denjoy integral of a suitable derivative \cite{dougherty1991complexity}. In analogy, \cite{StacsBournezGozzi2024} achieved the purpose of selecting the right class of ODEs to be used to obtain the same type of characterization of ordinal computing, but this time in the context of dynamical systems. This intuition is backed by highlighting substantial evidence regarding the unique definability properties of solutions of solvable systems. An extended version of \cite{StacsBournezGozzi2024}, with lengthier explanations and complete proofs, can be found in \cite{bournez2024solvable}. 

\paragraph{Contributions} 
In this article, we start from these premises of \cite{StacsBournezGozzi2024}  and fill several important gaps that were left open there, further fortifying the strength of the argument supporting solvable ODEs as the right candidates for an alternative description of transfinite computation. 

After refreshing the results from \cite{StacsBournezGozzi2024}, we conduct a formal analysis of real solvable functions describing their set descriptive complexity as a subset of the set of differentiable functions. 

More precisely, by introducing an ordinal rank within the class of solvable functions, we rigorously demonstrate that for each countable ordinal there exists at least one corresponding solvable dynamic, meaning that the whole class of solvable functions is fully populated and unbounded below the first uncountable ordinal. Notably, once again in accord with the results from \cite{dougherty1991complexity} on Denjoy integration, the latter implies that, in general, the operation of solving discontinuous ODEs is not Borel as soon as solvable functions are included in the picture, and therefore the totality of countable ordinals is necessary to perform such operation. To show this unboundedness result, we twist some of the techniques applied in \cite{westrick2014lightface} to differentiable functions and we construct a Borel function that maps each well-founded tree to one solvable function with solvable ranking linearly dependent on the limsup ranking of the corresponding tree. 

Another contribution of this paper is to expand the scope of the research field promoted by works such as \cite{kechris1986ranks}, \cite{Ra91}, and \cite{ki1997denjoy}, where the authors have analyzed the relation between many of the known rankings for differentiable functions. In this context, the rankings that have been discussed are the Kechris-Woodin ranking and the Denjoy ranking as in \cite{kechris1986ranks}, and the Zalcwasser ranking as in \cite{AK87}. In \cite{ki1997denjoy} is proved how the Denjoy and the Zalcwasser rankings are incomparable while being both dominated by the Kechris-Woodin ranking. 

In this article we manifest that restricting the focus to one-dimensional solvable functions provides interesting comparisons with the other rankings, proving that even when only differentiable functions are concerned, the solvable ranking stands out as more informative and strictly more complex than any of the other three rankings. More precisely we show that for each differentiable, real solvable function its solvable ranking is at least equal to its Kechris-Woodin ranking or greater. We argue that this result follows as a natural consequence of the fact that our solvable ranking is designed based on the transfinite method we developed in \cite{StacsBournezGozzi2024} for solving systems of ODEs, and therefore generalizes the other rankings which are more tailored over integration and differentiability properties of the functions considered. 

Finally, we show that the set of solvable functions is coanalytic as a subset of the space of real continuous functions and that the solvable ranking is a coanalytic norm on such a set. This appears as a clear difference with the Denjoy ranking, despite the clear similarities in the transfinite methodology that defines them. All these results just mentioned demonstrate the unique interesting properties of the ranking we introduce in this paper and outline the potential for its applications in describing the complexity of discontinuous dynamical systems in higher dimensions as well.   

\subsection{Structure of the paper}

The paper is organized in the following manner: Section \ref{sec:prel} presents the preliminaries for the article, with a subsection for the notation, one for the hyperarithmetical hierarchy, trees and coanalytic norms and one for ordinary differential equations; Section \ref{sec:solvable} introduces the class of solvable functions and describes the properties of solutions of solvable systems while presenting a few notable examples; Section \ref{sec:main} contains the main results of the paper: it introduces the solvable ranking, proves its most interesting properties and discusses the comparison with other existing rankings for differentiable functions; finally, Section \ref{sec:conclusions} concludes the paper, mentioning possible directions for future work.

\section{Preliminaries}
\label{sec:prel}

\subsection{Notation}

We introduce the notation and the main definitions used throughout this work. The standard symbols $\mathbb{N}$, $\mathbb{N}_0$, $\mathbb{R}$ and $\mathbb{Q}$ stand for the set of natural, natural including zero, real and rational numbers respectively while $\mathbb{R}^+$ and $\mathbb{Q}^+$ represent the positive reals and the positive rationals. When making use of the norm operator on Euclidean spaces, we always consider Euclidean norms if not specified otherwise. We use $\langle n_1,\hdots,n_k \rangle$ to denote a single integer which represents the tuple $( n_1,\hdots,n_k)$ according to some standard computable encoding. We call a space Polish space if it is a separable, completely metrizable topological space. We call $C([0,1])$ the Polish space of continuous real functions with the uniform metric. We refer to a compact domain of a topological space as a nonempty connected compact subset of such space. Given a metric space $X$ we indicate with the notation $d_X$ the distance function in such space and with the notation $B_X(x, \delta)$ the open ball centered in $x \in X$ with radius $\delta >0$. By default we describe as open \emph{rational ball} or as open \emph{rational box} an open ball or box with rational parameters. Precisely, an open rational box $B$ is a set of the form $(a_1, b_1) \times \hdots \times (a_r, b_r) \subset \mathbb{R}^r$ for some $r \in \mathbb{N}$ where $a_i, b_i \in \mathbb{Q}$ for $i = 1, \hdots, r$. We indicate with the notations $\operatorname{diam}(B)$ and $\operatorname{rad}(B)$ its diameter and its radius. Moreover, given a function $f: [a,b] \to \mathbb{R}^r$ for some $a,b \in \mathbb{R}, a<b$ and some $r \in \mathbb{N}$ we indicate with the notation $f': [a,b] \to \mathbb{R}^r$ the derivative of such function, where the derivative on the extremes $a$ and $b$ is defined as the limit on the left and on the right respectively. For any function $f \in C([0,1])$ and any $[a, b] \subseteq [0,1]$ we write $f[a, b]$ to denote the function which is identically 0 outside of $[a, b]$, and instead values $f[a, b](x)=(b-a) f\left(\frac{x-a}{b-a}\right)$ for all $x \in[a, b]$. Note that if $f$ is continuous and $f(0)=f(1)=0$, then $f[a, b]$ is continuous; it is computable when $f, a$, and $b$ are and differentiable when $f$ is differentiable and $f^{\prime}(0)=f^{\prime}(1)=0$. Given a function $f: X \to Y$ and a set $K \subseteq X$ we indicate with the notation $f \restriction_K$ the restriction of function $f$ to the set $K$, i.e.\ $f \restriction_K$ is the function from $K$ to $Y$ defined as $f \restriction_K (x)= f(x)$. If $A$ and $B$ are two sets, we refer to the set difference operation using the symbol $A \setminus B$ and we indicate with the notation $A+B$ the Minkowski sum of set $A$ with set $B$. The expression $\operatorname{cl}(A)$ indicates the closure of $A$, $\operatorname{int}(A)$ the interior of $A$ and $A^c$ its complement. The notation $\emptyset$ stands for the empty set, while the notation $\omega_1$ stands for the first uncountable ordinal number. Given a property of a function $f: X \to Y$, we say that this property is satisfied \emph{almost everywhere} if the property is satisfied on $X \setminus D$, where $D$ is a set with Lebesgue measure equal to zero. We use the symbol $\phi_e$ to denote to the eth Turing functional acoordinal to some G\"{o}del enumeration. One reducibility is denoted by $\leq_1$. Given a set $X \subset \mathbb{N}$, the jump of $X$ is written $X'$ and the nth jump of $X$ is written $X^{(n)}$.

\subsection{Hyperarithmetical hierarchy, trees and coanalytic norms}

We describe the process of \emph{transfinite recursion} as the process that for each ordinal $\alpha$ associates with $\alpha$ an object that is described in terms of objects already associated with ordinals $\beta < \alpha$ \cite{SA17}. We use the expression \emph{transfinite recursion up to $\alpha$} if the process associates an object for all ordinals $\beta <\alpha$.

We assume our reader to have some basic familiarity with the boldface and lightface hierarchies: see e.g \cite{moschovakis2009descriptive,SA17}. 

To recall the hyperarithmetical hierarchy we make use of the notion of constructive ordinals, which we define by means of the concept of Kleene's $\mathcal{O}$.

\begin{definition}[Kleene's $\mathcal{O}$]
Let the relation $<_{\mathcal{O}}$ on the natural numbers be the least relation closed under the following properties: 
\begin{itemize}
\item $1 <_{\mathcal{O}} 2$
\item If $a <_{\mathcal{O}} b$ then $b <_{\mathcal{O}} 2^b$
\item If funtion $\phi_e$ is total and $\phi_e(n) <_{\mathcal{O}} \phi_e(n+1)$ for all $n \in \mathbb{N}$ then $\phi_e(n) <_{\mathcal{O}} 3 \cdot 5^e$ for all $n \in \mathbb{N}$ 
\item If $a <_{\mathcal{O}} b$ and $b <_{\mathcal{O}} c$ then $a <_{\mathcal{O}} c$
\end{itemize}
we call the field of this relation Kleene's $\mathcal{O}$. 
\end{definition}

This field has many interesting properties, such as being $\Pi_1^1$ complete, and relation $<_{\mathcal{O}}$ is well-founded. Moreover, for every $a \in \mathcal{O}$ the set defined as $\{ b <_{\mathcal{O}} a\}$ is computably enumerable and well ordered. This implies that for every $a \in \mathcal{O}$ there is a well defined ordinal $\left\vert a \right\vert_{\mathcal{O}}$ such that $\left\vert a \right\vert_{\mathcal{O}}$ is the order type of the set $\{ b <_{\mathcal{O}} a\}$. We call the natural number $a \in \mathcal{O}$ the \emph{ordinal notation} for the ordinal $\left\vert a \right\vert_{\mathcal{O}}$. 

\begin{definition}[Constructive ordinals]
We call an ordinal number \emph{constructive} if it has an ordinal notation in Kleene's $\mathcal{O}$.
\end{definition}

It is clear that there are only countably many constructive ordinals, while each constructive ordinal greater than $\omega$ admits infinitely many possible ordinal notations. We indicate the first nonconstructive ordinal with the symbol $\omega_1^{CK}$, which stands for \emph{the Church-Kleene ordinal}. 

We now recall the definition of the arithmetical hierarchy. We say that a set $X$ is $\Sigma_n$ ($\Pi_n$ respectively) if $X \leq_1 \emptyset^{(n)}$ (respectively $X \leq_1 (\emptyset^{(n)})^c$). The arithmetical hierarchy can be extended to transfinite levels for all constructive ordinals, generating in this way the hyperarithmetical hierarchy. Define the a sequence of sets, one for each ordinal notation in Kleene's $\mathcal{O}$, in the following way: $H_1= \emptyset$, $H_{2^b}=(H_b)^{'}$ and $ H_{3 \cdot 5^{e}}=\{ \langle x, n \rangle : x \in H_{\phi_e(n)} \}$. Then we have: 

\begin{definition}[Hyperarithmetical hierarchy]
Let $\alpha < \omega_1^{CK}$ be an infinite ordinal and let $X \subset \mathbb{N}$. Then $X$ is $\Sigma_\alpha$ if  $X \leq_1 H_{2^a}$ for any $a$ such that $\left\vert a \right\vert_{\mathcal{O}}$. The definition for $\Pi_\alpha$ is done in the same way. 
\end{definition}

We can then extend the notion of Turing jumps to every ordinal less than $\omega_1^{CK}$. This is done in the following way: we first fix a particular (but arbitrary) path $\mathcal{P}$ through $\mathcal{O}$ (A path $\mathcal{P}$ through $\mathcal{O}$ is a subset $ \mathcal{P} \subseteq \mathcal{O}$ such that it is $<_{\mathcal{O}}$ linearly ordered and contains an ordinal notation for each $\alpha < \omega_1^{CK}$). We then define, for all $\alpha < \omega_1^{CK}$, the $\alpha$-Turing jump as the set $\emptyset^{(\alpha)}= H_a$, where $a$ is the unique $a \in \mathcal{P}$ such that $\left\vert a \right\vert_{\mathcal{O}}=\alpha$. 

The hyperarithmetical hierarchy can be considered naturally also in the context of real numbers. Indeed, a real number $x$ is called hyperarithmetical if the set of rationals $\{ q \in \mathbb{Q} : q < x \}$ is hyperarithmetical.

A rank function or norm on a set $P$ is just a map $\varphi: P \rightarrow O R D$, from $P$ into the class of ordinals ORD. This induces a prewellorlering $\leqslant_{\varphi}$ on $P$ defined by:

\begin{equation*}
x \leqslant_{\varphi} y \Longleftrightarrow \varphi(x) \leqslant \varphi(y) 
\end{equation*} 

Two norms $\varphi, \varphi^{\prime}$ on $P$ are equivalent if they induce the same prewellorlerings $\leqslant_{\varphi}=\leqslant_{\varphi^{\prime}}$. We are mainly interested in special norms on $\boldsymbol{\Pi}_1^{1}$ subsets of Polish spaces, these are called \emph{coanalytic norms}.

\begin{definition}[$\boldsymbol{\Pi}_1^1$ norm] 
Given $X$ a Polish space and a $\boldsymbol{\Pi}_1^1$ subset $P$ of $X$ we say that a norm $\varphi: P \rightarrow O R D$ is a $\boldsymbol{\Pi}_1^1$ norm (on $P$) if there is a $\boldsymbol{\Sigma}_1^1$ subset $R$ of $X^2$ and a $\boldsymbol{\Pi}_1^1$ subset $Q$ of $X^2$ (i.e., $R, Q$ are relations on $X$ ) such that:
\begin{equation*}
y \in P \Longrightarrow[(x \in P \land \varphi(x) \leqslant \varphi(y)) \Longleftrightarrow R(x, y) \Longleftrightarrow Q(x, y)]
\end{equation*}
\end{definition}

It is well known that every $\boldsymbol{\Pi}_1^1$ norm is equivalent with one which takes values in $\omega_1$. Indeed, the following proposition holds \cite{kechris2012classical}.

\begin{proposition}
\label{prop:coananorm}

Every $\boldsymbol{\Pi}_1^1$ set $A$ in a Polish space admits a $\boldsymbol{\Pi}_1^1$ norm $\varphi: A \rightarrow \omega_1$. If $\varphi: A \rightarrow \alpha$ is a $\boldsymbol{\Pi}_1^1$ norm, then for each $\xi<\alpha$ let
$$
A_{\xi}=\{x \in A: \varphi(x) \leq \xi\} .
$$

Then $A_{\xi}$ is Borel, $A_{\xi} \subseteq A_\eta$ if $\xi \leq \eta$, and $A=\bigcup_{\xi<\alpha} A_{\xi}$. In particular, every $\boldsymbol{\Pi}_1^1$ set is the union of Borel sets. Also, if $\varphi: A \rightarrow \omega_1$ is a $\boldsymbol{\Pi}_1^1$ norm on a $\boldsymbol{\Pi}_1^1$ but not Borel set, then $\sup \{\varphi(x): x \in A\}=\omega_1$.

\end{proposition}

We define trees in the usual way as nonempty subsets of $\mathbb{N}_0^{<\mathbb{N}_0}$ closed under initial sequences. If $m$ is a natural number and $\sigma= \langle n_1, \hdots, n_k \rangle$ is a sequence, $m^\frown \sigma$ and $\sigma ^\frown m$ denote $\langle m, n_1, \hdots, n_k \rangle$ and $\langle n_1, \hdots, n_k, m\rangle$ respectively. If $T$ is a tree, let $T_n$ denote $\{ \sigma : n^\frown \sigma \in T \}$, i.e.\ the nth subtree of $T$. If $T$ is well-founded, $\vert T \vert$ denotes its rank. The set of all well-founded trees, often indicated with $WF$, is a $\boldsymbol{\Pi}_1^{1}$ complete subset of $\mathbb{N}_0^{<\mathbb{N}_0}$. A $\boldsymbol{\Pi}_1^{1}$ subset $A$ of a Polish space $X$ is said to be complete if
for any Polish space $Y$ and any $\boldsymbol{\Pi}_1^{1}$ subset $B$ of $Y$ there is a Borel measurable function $f: Y \to X$ such that $B = f^{-1}(A)$. 

Another notion of ranking for well-founded trees have been introduced in \cite{Wes14} in order to define a $\boldsymbol{\Pi^1_1}$ norm that could resemble the structure of the Cantor-Bendixson derivative. 

\begin{definition}[Limsup rank]
Let $T$ be a well-founded tree, define the \emph{limsup rank} of $T$ as:
\begin{equation*}
\vert T \vert_{ls} = \max \left( \sup_n \vert T_n \vert_{ls} , \left( \limsup_n \vert T_n \vert_{ls}\right) +1 \right)
\end{equation*}
if $T$ is nonempty, and $\vert T \vert_{ls}=0$ if $T$ is empty. 
\end{definition}

Note that reordering the subtrees does not change the limsup rank of the tree. A node can have a rank higher than all its children in one of two situations: either there is no child of maximal rank, or there are infinitely many maximal rank children. Moreover, in \cite{Wes14} it is proved that the limsup rank of any well-founded tree is always a successor ordinal. 

Abstractly, from the fact that $WF$ is a $\boldsymbol{\Pi}_1^{1}$ complete set, it is always possible to assign, given a Polish space $X$ and a $\boldsymbol{\Pi}_1^{1}$ set $A \subseteq X$, a tree $T_x$ to each $x \in X$ such that the function that maps $x \to T_x$ is Borel and $x \in A$ if and only if $T_x \in WF$. Then, the function that maps $x \to \left\vert T_x \right\vert$ (or $x \to \left\vert T_x \right\vert_{ls}$) is a $\boldsymbol{\Pi}_1^{1}$ norm on $A$. 

\subsection{IVPs with ordinary differential equations}

We consider dynamical systems whose evolution is described by ordinary differential equations. Consider an interval $[a,b] \subset \mathbb{R}$, a closed domain $E \subset \mathbb{R}^r$ for some $r \in \mathbb{N}$, a point $y_0 \in E$ and a function $f: E \to \mathbb{R}^r$ such that the dynamical system:

\begin{equation}
\label{eq:pb}
\begin{cases}
y'(t)= f(y(t))\\
y(a)=y_0
\end{cases}
\end{equation}

has at least one solution $y:[a,b] \to \mathbb{R}^r$ with $y([a,b]) \subset E$. Given $y_0$ and $f$, the problem of obtaining one solution in such a setting is called an \emph{Initial Value Problem (IVP)}. The condition $y(a)=y_0$ (or, in short, just the point $y_0$) is referred to as the initial condition of the problem, and function $f$ is referred to as the right-hand term of the problem. In general, there can be multiple valid solutions for the same IVP. Nonetheless, in this work we consider only IVPs whose solution is uniquely defined. Therefore, we refer to function  $y:[a,b] \to E$ satisfying Equation \eqref{eq:pb} as the solution of the problem. 

For IVPs with continuous right-hand terms, several different methods have been developed to obtain the solution analytically, all based on the concept of constructing a sequence of continuous functions converging to the solution. Once it is known that the solution is unique, every sequence considered in each of these methods can be shown to converge to the unique solution. The authors of \cite{collins2009effective} have demonstrated that continuity of the right-hand term together with uniqueness corresponds to a sufficient condition to achieve computability for the solution, where computability has to be intended as the possibility of computing any approximation of such function uniformly on the precision input, as usual in the sense of computable analysis. The procedure is detailed in the description of their so-called \emph{ten thousand monkeys} algorithm. The idea of this algorithm is to exploit the hypotheses of unicity and continuity for enclosing the solution into covers of arbitrarily close rational boxes in $E$. 

Making use of a similar technique, our analysis in \cite{StacsBournezGozzi2024} proves that a weaker condition is sufficient to analytically produce the unique solution even in the presence of highly discontinuous behaviors for the right-hand terms considered. We showed that this is possible as long as we are willing to pay the price of allowing a transfinite procedure into the picture.  This approach differs from the general treatment within the theory of discontinuous ODEs \cite{FilippovBook,aubin2012differential,deimling2011multivalued} where the notion of solution varies greatly depending on the main goal of the methodology involved. Indeed, the almost everywhere equality in Equation \eqref{eq:pb} typical of most approaches allows for more freedom on an alternative interpretation of the solution, according to which behavior is chosen to be desirable on the discontinuity points. In this context instead, ambiguities concerning the concept of solution are not possible, since we require function $y$ in Equation \eqref{eq:pb} to be unique and everywhere defined. 

\section{Solvable functions and examples}
\label{sec:solvable}

In this section we briefly describe the sufficient condition outlined in \cite{StacsBournezGozzi2024} for the right-hand terms of the IVPs considered and present a few examples. The sufficient condition is based on the definition of solvable functions, which intuitively can be described as functions that allow for a precise stratification on their degree of discontinuity. We can quantify this formally by first introducing the following definition:

\begin{definition}[Sequence of $f$-removed sets on $E$]
\label{def:removed}

For some $r,m \in \mathbb{N}$, consider a compact domain $E \subset \mathbb{R}^r$ and a function $f: E \to \mathbb{R}^m$. Let $\{ E_\alpha\}_{\alpha< \omega_1}$ be a transfinite sequence of sets and $\{ f_\alpha \}_{\alpha}$  a transfinite sequence of functions such that $f_\alpha = f \restriction_{E_{\alpha}}: E_{\alpha} \to \mathbb{R}^m$ defined as following:

\begin{itemize}
\item Let $E_0=E$ 
\item For every $\alpha$ let $E_{\alpha+1}=  D_{ f_{\alpha}}$
\item For every $\alpha$ limit ordinal, let $E_\alpha= \displaystyle\cap_{\beta<\alpha} E_\beta $
\end{itemize}

we call the sequence $\{ E_\alpha\}_{\alpha}$ the \emph{sequence of $f$-removed sets} on $E$.

\end{definition}

We remark that since functions in the sequence $\{ f_\alpha \}_{\alpha}$ above are allowed to be defined over disconnected sets, the notion of continuity in the above definition has to be intended for the induced topology relative to $E_{\alpha}$ as a subset of $\mathbb{R}^r$. Moreover, note that it follows from the definition of the sequence that such sequence is decreasing, meaning that $E_\delta \subseteq E_\gamma$ if $\delta > \gamma$ for every $E_\delta$ and $E_\gamma$ in the sequence. 

\begin{definition}[Solvable function] \label{def:solvable}
For some $r,m \in \mathbb{N}$, consider a compact domain $E \subset \mathbb{R}^r$ and a function $f: E \to \mathbb{R}^m$. We say that $f$ is \emph{solvable} if it is a function of class Baire one such that for every closed set $K \subseteq E$ the set of discontinuity points of the restriction $f \restriction_K$ is a closed set.
\end{definition}

We say that a dynamical system or an ordinary differential equation is \emph{solvable} when the right-hand term of the ordinary differential equation is a solvable function. Solvable functions have the following property \cite{StacsBournezGozzi2024}:

\begin{theorem}
\label{thm:lastord}
For some $r,m \in \mathbb{N}$, consider a compact domain $E \subset \mathbb{R}^r$ and a function $f: E \to \mathbb{R}^m$. If $f$ is solvable, then there exists an ordinal $\alpha<\omega_1$ such that $E_{\alpha}= \emptyset$.
\end{theorem}

It is clear from the above theorem that the definition of solvable functions provides an intuitive way to rank \emph{solvable} systems, where the ranking naturally corresponds to the first countable ordinal that leads to the empty set in the sequence of removed sets. A precise notion for this ranking, as well as a thorough description of its properties, is given in Section \ref{sec:main} and represents the main scope of this paper. 

Solvable right-hand terms together with unicity represent a sufficient condition for solving analytically discontinuous IVPs using an appropriate search method within the domain. The search method is based on the ordinal classification of the discontinuity for the right-hand term according to the theorem above, and therefore one such method is defined uniformly for each countable ordinal. Tuples that are labeled as valid by the method are then used to identify the correct sequence of open rational boxes that approximates the solution. It follows the formal definition of the search method: 

\begin{definition}[$(\alpha)$Monkeys approach]
\label{def:alphamonkeys}

Consider an interval $[a,b] \subset \mathbb{R}$, a domain $E \subset \mathbb{R}^r$ for some $r \in \mathbb{N}$ and a right-hand term $f: E \to \mathbb{R}^r$ for an ODE of the form of the one in Equation \eqref{eq:pb} with initial condition $y_0$. Let $\{ E_\gamma \}_{\gamma< \omega_1}$ be the sequence of $f$-removed sets on $E$ and let $E_\alpha$ be one set in the sequence for some $\alpha< \omega_1$. We call the \emph{$(\alpha)$Monkeys approach for $(f, y_0)$} the following method: consider all tuples of the form $\left(X_{i, \beta, j }, h_{i, \beta, j }, B_{i, \beta, j }, C_{i, \beta, j }, Y_{i, \beta, j }\right)$ for $i=0, \ldots, l-1$, $\beta < \alpha$, $j=1, \ldots, m_{i,\beta}$, where $h_{i, \beta, j } \in \mathbb{Q}^+$, $l, m_i \in \mathbb{N}$ and $X_{i, \beta, j }$, $B_{i, \beta, j }$, $C_{i, \beta, j }$ and $Y_{i, \beta, j }$ are open rational boxes in $E$. A tuple is said to be valid if $y_0 \in \bigcup_{\beta, j} X_{0, \beta, j }$ and for all $i=0, \ldots, l-1$, $\beta < \alpha$, $j=1, \ldots, m_{i,\beta}$ we have:

\begin{enumerate}
\item Either $\left( B_{i, \beta, j }= \emptyset \right)$ or $\left( \operatorname{cl}(B_{i, \beta, j }) \cap E_\beta \neq \emptyset \; \text{and} \; \operatorname{cl}(B_{i, \beta, j }) \cap E_{\beta+1} = \emptyset  \right)$
\item $ f \restriction_{E_\beta} \left( \operatorname{cl}(B_{i, \beta, j }) \right) \subset C_{i, \beta, j }$;
\item $ X_{i, \beta, j } \cup  Y_{i, \beta, j } \subset B_{i, \beta, j }$;
\item $ X_{i, \beta, j } + h_{i, \beta, j } C_{i, \beta, j }\subset Y_{i, \beta, j }$;
\item $ \bigcup_{\beta, j} Y_{i, \beta, j } \subset \bigcup_{\beta,j} X_{i+1,\beta, j}$;
\end{enumerate}
\end{definition}

We can now present the main theorem concerning solutions of solvable IVPs that support the claim above. A detailed description of the proof of the theorem, as well as an exhaustive explanation of the roles that Definitions \ref{def:removed}, \ref{def:solvable} and \ref{def:alphamonkeys} play in such proof, can be found in \cite{bournez2024solvable}.

\begin{theorem}[{\cite{bournez2024solvable}}]
\label{thm:main}
Consider a closed interval, a compact domain $E \subset \mathbb{R}^r$ for some $r \in \mathbb{N}$ and a function $f: E \to E$ such that, given an initial condition, the IVP of the form of Equation \eqref{eq:pb} with right-hand term $f$ has a unique solution on the interval. 

If $f$ is a solvable function, then we can obtain the solution analytically via transfinite recursion up to an ordinal $\alpha$ such that $\alpha < \omega_1$. 
\end{theorem}

\subsection{Examples}
\label{sec:example}

We present here a few of the examples described in \cite{StacsBournezGozzi2024} to give the reader a clearer picture of solvable systems. 

\begin{example}[A discontinuous IVP with a unique solution]
\label{ex:firstex}
As the simplest case of discontinuous IVP, we consider the following example: let $E= [-5,5] \times [-15,15]$ and define the function $f: E \to E$ as $f(x,z)=(1,2x \sin{\frac{1}{x}} - \cos{\frac{1}{x}})$ if $x\neq 0$ and $f(0,z)=(1,0)$. It is easy to see that $f$ is a function of class Baire one, i.e.\ it is the pointwise limit of a sequence of continuous functions. Note also that in this case, the set of discontinuity points of the function $f$ on $E$ is the closed set $E_1= \{ (0,z)$ for $z \in [-15,15] \}$.  Then consider the following IVP, with $y: [-2,2] \to \mathbb{R}^2$ and $y_0=(-3, 9 \sin(-\frac{1}{3}))$:

\begin{equation}
\label{eq:derivODE2dimens}
\begin{cases}
y'(t)= f(y(t))\\
y(-2)= y_0
\end{cases}
\end{equation}

It is easy to verify that the solution of such a system is unique, and it is for the first component: $y_1(t)= t- 1$, and the second component $y_2 (t)= (t-1)^2 \sin(\frac{1}{t-1})$ for $t \neq 1$ and $y_2(1)=0$. Therefore the solution $y: [-2,2] \to \mathbb{R}^2$ is differentiable and can be expressed as the unique solution of the IVP above with right-hand side $f$ discontinuous on $E$. Note that the only discontinuity of $f$ encountered by the solution is the point $(0,0)$, i.e.\ $E_1 \cap y([-2,2])= (0,0)$. Moreover, it is easy to see that the right-hand term $f$ is a solvable function. Indeed, $f$ is trivially a function of class Baire one, while $f \restriction_{E_1}$ is a constant function identically equal to $(1,0)$ and therefore continuous everywhere in its domain. This also implies $E_2=\emptyset$, where $E_2$ is the second set in the sequence of $f$-removed sets on $E$. 
\end{example}

The construction of this simple example (and of the solution of this classical exercise) is based on the well-known fact that the real function $f(x)=x^2 \sin (\frac{1}{x})$ if $x\neq 0$ and $f(0)=0$ is differentiable over $[0,1]$ and its derivative is bounded and discontinuous in $0$. Moreover, we avoided some problems that arise for one-dimensional ODEs with null derivative by introducing a \emph{time} variable $y_1$ whose role is to prevent the system from stalling and ensure the unicity of the solution. 

The concept behind such an example can be easily generalized by making use of differentiable functions with more and more complicated derivatives. It is indeed possible to construct in this way examples with right-hand terms discontinuous over sets that are homeomorphic to the Cantor set, as well as examples with $E_2 \neq \emptyset$. 

Another interesting example showcases the relation between ordinal Turing computations and solvable IVPs by designing one such IVP whose solution encodes, at an integer time, the halting set for Turing machines. Hence the following example represents the first milestone for a possible equivalent formulation of hyperarithmetical reals using solvable IVPs, in the same sense as what has been done in \cite{dougherty1991complexity} for the Denjoy integral. 

\begin{example}[A solvable IVP encoding the halting set] \label{ex:secondex}

Given a G\"{o}del enumeration of Turing machines, we define the halting problem as the problem of deciding the halting set $H=\{ e : M_e (e) \perp \}$ where $ M_e (e) \perp$ means that machine represented by natural $e$ halts on input $e$. We consider a one-to-one total computable function over the naturals $h:\mathbb{N} \to \mathbb{N}$ that enumerates such a set. It is known that any such function enumerating a noncomputable set naturally generates a noncomputable real number \cite{PORI17}. The following definition expresses this more formally:

\begin{definition}
\label{def:noncompu}
Let $h:\mathbb{N}_0 \to \mathbb{N}$ be a one-to-one computable function such that $h(i)>0$ for all $i \in \mathbb{N}_0$ and such that it enumerates a noncomputable set $A$. Then the real number $\mu$ defined as $\mu= \sum_{i=0}^{\infty} 2^{- h(i)}$ is noncomputable.
\end{definition}

According to this notion of encoding, the existence of an example with the properties mentioned above is expressed by the following theorem, whose complete proof can be found in \cite{bournez2024solvable}:

\begin{theorem}[{\cite{bournez2024solvable}}]
\label{thm:HaltingP}
Let $E=[-5,5] \times [-5,5]$ and let $\mu$ be a real encoding of the halting set for Turing machines. There exists a solvable IVP with unique solution $y: [0,5] \to E$, rational initial condition, and right-hand term computable everywhere on $E$ except a straight-line, that satisfies $y_2 (5)=\mu$. 
\end{theorem}

\end{example}
\section{Set descriptive complexity of solvable functions}
\label{sec:main}

The goal of this section is to investigate the properties of the class of solvable functions by introducing a natural ranking based on the complexity of their discontinuities. For simplicity of presentation, we restrict the analysis to the solutions of solvable systems with domain $[0,1]$, but the treatment can be easily extended to any real closed interval. 

For formalizing the definition of the ranking, we are immediately faced with important choices to make concerning the base space to consider. Indeed, since in this context we are dealing with dynamical systems, it is not immediately clear which one should be the base space in which to inscribe the ranking. This is because the right-hand term, the initial condition, and the solution all live in different spaces, i.e.\ the first one is a discontinuous function, the second one a vector, and the third one a differentiable function. One possible interpretation could be to choose as the base space the product of all such involved spaces, in order to have a full complete description of the system. Nonetheless, we believe that such a description is redundant, and instead we prefer to focus only on the solutions to also preserve the similarities with differentiability rankings that lead to interesting comparisons. Our choice has been an operational one, we decided to select the search method we proposed to find the solution (i.e.\ the $(\alpha)$Monkeys approach as in Definition \ref{def:alphamonkeys}) as the inner core criterium for describing the ranking of such solutions. More precisely, the relation between the solution and the sequence of removed sets is the key element that provides complexity classification. This is expressed by the following definition: 

\begin{definition}[Solvable ranking]
\label{def:solvablerank}
Consider a compact domain $E \subset \mathbb{R}^r$ for some $r \in \mathbb{N}$ and a differentiable function $y:[0,1] \to E$ with solvable derivative $y'$. Let $\{ E_\gamma \}_{\gamma< \omega_1}$ to be the sequence of $y'$-removed sets on $[0,1]$. We define the \emph{solvable ranking of $y$}, that we indicate with $\vert y \vert_{SV}$, as $\vert y \vert_{SV}= \alpha$, where $\alpha$ is the least ordinal such that $E_\alpha = \emptyset$.  
\end{definition}

Several observations and discussions are due at this point to better understand this ranking. The first observation is that the above ranking is well-defined. Indeed, if $y'$ is a solvable function, it follows from Theorem \ref{thm:lastord} that there exists an ordinal $\alpha < \omega_1$ in the sequence of $y'$-removed sets on $[0,1]$ such that $E_\alpha = \emptyset$. For all $\gamma < \omega_1$, we define $SV_{\gamma}= \bigl\{ y \in \mathcal{D}: \left\vert y \right\vert_{SV} \leq \gamma \bigr\}$.

The idea and justification behind the above definition can be described briefly in an intuitive manner. The above definition refers only to differentiable functions with solvable derivatives and this is because for these functions we know that we can design a dynamical system based on an IVP for which $y$ will be the solution and for which such solution can be obtained with our method as in Theorem \ref{thm:main}. Therefore in this case we can rank the complexity depending on the number of steps required by such a method. Note that this description is independent of the actual formalization of the dynamical system considered, the right-hand term involved, or the initial condition chosen within the domain. The reason why we chose to provide this independence to our ranking is better explained by the following example: 

\begin{example}[A simple solution within a discontinuous dynamical system]

The idea of this trivial example is to imagine a situation in which the solution of the dynamical system is continuously differentiable, i.e.\ $y \in C^1([0,1])$, but lives in a domain $E$ where the right-hand term $f$ is discontinuous on any ball intersecting $y([0,1])$. This would imply that any application of our analytical method based on Definition \ref{def:alphamonkeys} would require at least one step to obtain the solution, despite it being continuously differentiable. Nonetheless, it is clear that by considering a restriction of the domain, or by carefully redefining the right-hand term outside the right trajectory of the evolution, the solution is easily retrievable from the initial condition. Therefore, this is an argument supporting the idea that the ranking should be independent from the particular domain or right-hand term, but should instead focus on the solution itself, or in other words, should focus only on the \emph{simplest possible} dynamical system that could present such solution.

 Let $E= [0,1] \times [0,1]$, consider an IVP on the interval $[0,1]$ of the form of Equation \eqref{eq:pb} where $y(0)=(1/2,0)$ and the right hand term $f:E \to E$ is defined in the following way: 

\begin{equation}
\begin{cases}
f(x,z)= (0,1) \text{  if  } x=1/2  \\
f(x,z)= (x,1) \text{ otherwise }\\
\end{cases}
\end{equation}

It is immediate that there exists a unique solution, i.e.\ $y(t)=(1/2,t)$ for all $t \in [0,1]$, and that such solution is continuously differentiable. Nonetheless, the right-hand term $f$, which is bounded and solvable, is discontinuous on the line $\{ (1/2,z) : z \in [0,1]\}$. This implies that the method that obtains the solution relies on the application of the $(1)$Monkeys approach. Therefore it follows that if we were to assign the solvable ranking to $y$ and $f$ only depending on the complexity of the analytical method we described, such a solution within the context of this dynamical system would have a ranking equal to $1$ despite having a continuous derivative. 
\end{example}

\subsection{Comparison with differentiability rankings}

We start by introducing the Kechris-Woodin ranking \cite{kechris1986ranks} for differentiable functions over $[0,1]$ that offers an interesting comparison with the solvable ranking in the one-dimensional case. 

Given $y \in C([0,1])$, for any two points $x<z$ we indicate the difference quotient $\Delta_y (x,z)$ as:

\begin{equation*}
\Delta_y (x,z)= \frac{y(x)-y(z)}{x-z}
\end{equation*}

the difference quotient allows us to define a derivative operator which leads to the construction of the Kechris-Woodin ranking. 

\begin{definition}[Kechris-Woodin ranking]

Let $y \in C([0,1])$, $\epsilon>0$ and $P \subset [0,1]$ closed. Define the set:

\begin{equation*}
\begin{split}
P'_{\epsilon, y}= \{ x \in P : \forall \delta >0 \exists p < q , r < s \in B_{[0,1]}(x, \delta) \cap \mathbb{Q} \; \text{with} \\
 [p,q] \cap [r,s] \cap P \neq \emptyset \; \text{and} \; \vert \Delta_y (p,q) - \Delta_y (r,s) \vert \geq \epsilon \}
\end{split}
\end{equation*} 

Since this set is closed, we construct the following inductive hierarchy: 

\begin{itemize}
\item  $P^{0}_{\epsilon,y}=[0,1]$ 
\item  $P^{\alpha+1}_{\epsilon,y}= (P^{\alpha}_{\epsilon,y})'_{\epsilon,y} $
\item  For $\alpha$ limit ordinal, $P^{\alpha}_{\epsilon,y}= \displaystyle\cap_{\beta< \alpha} (P^{\beta}_{\epsilon,y})'_{\epsilon,y}$
\end{itemize}

We call the Kechris-Woodin ranking of $y$, that we indicate with $\vert y \vert_{KW}$, the least ordinal $\alpha$ such that $(\forall \epsilon >0) P^{\alpha}_{\epsilon,y} = \emptyset$. 

\end{definition}

It is known that the set of differentiable functions from $[0,1]$ to $[0,1]$, that we call $\mathcal{D}$, is a $\boldsymbol{\Pi}_1^{1}$ subset of $C([0,1])$. Moreover, it is also $\boldsymbol{\Pi}_1^{1}$ complete \cite{Maz36}. Consequently $\mathcal{D}$ is not a Borel subset of $C([0,1])$. By looking at the above ranking as a function from $\mathcal{D}$ to the class of ordinals, the authors proved in \cite{kechris1986ranks} that such ranking is a $\boldsymbol{\Pi}_1^{1}$ norm. If for all $\alpha< \omega_1$ we define $KW_{\alpha}= \bigl\{ y \in \mathcal{D}: \left\vert y \right\vert_{KW} \leq \alpha \bigr\}$ then the latter implies due to Proposition \ref{prop:coananorm} that $\bigcup_{\alpha< \omega_1} KW_{\alpha} = \mathcal{D}$ and that $\sup \bigl\{ \left\vert y \right\vert_{KW} : y \in \mathcal{D} \bigr\}= \omega_1$. This is of course not the only possible $\boldsymbol{\Pi}_1^{1}$ norm on $\mathcal{D}$, and interesting comparisons between different $\boldsymbol{\Pi}_1^{1}$ norms can be found in \cite{Ra91}. 

Furthermore, it is immediately clear from the above definition that the inductive hierarchy of closed sets used for defining this ranking has many analogies in structure with the sequence of removed sets presented in Definition \ref{def:removed}. Indeed, we report here some properties of the Kechris-Woodin ranking proved in the original article that have clear immediate assonance within the context of some of the examples of solvable functions previously discussed in Section \ref{sec:example}. 

\begin{proposition}
Let $y:[0,1] \to [0,1]$ be a differentiable function. Then $\vert y \vert_{KW}= 1$ if and only if  
$y$ is continuously differentiable.
\end{proposition}

\begin{proposition}
Let $y:[0,1] \to [0,1]$ be the function defined as $y(0)=0$ and $y(x)= x^2\sin(1/x)$ otherwise. Then $\vert y \vert_{KW}= 2$.
\end{proposition} 

\begin{proposition}
For each ordinal $\alpha < \omega_1$ there exists a differentiable function $y:[0,1] \to [0,1]$ with bounded derivative such that $\vert y \vert_{KW}= \alpha$.
\end{proposition} 

In particular, the last remark easily relativizes to the constructive ordinals. 

The first noticeable difference between the two rankings just defined is that the definition of the solvable ranking concerns vector-valued functions. Nonetheless, since we decided to restrict the scope of the solvable ranking only to differentiable functions, it is interesting to compare it with the Kechris-Woodin ranking in the one-dimensional case. Moreover, the solvable ranking only refers to the class of differentiable functions with bounded, solvable derivatives, which, as proved from a counterexample in \cite{bournez2024solvable}, is a proper subset of $\mathcal{D}$. In other words, there exist functions in $\mathcal{D}$ with bounded derivatives that do not have a solvable derivative. Therefore, for the sake of properly investigating the comparison between the two different rankings, we focus for the rest of the section on the set of functions belonging to a class that we indicate with the symbol $SV$, i.e.\ the class of functions in $\mathcal{D}$ with bounded, solvable derivative. Hence, for all $\alpha< \omega_1$, we define $KW^{*}_{\alpha}= \bigl\{ y \in SV : \left\vert y \right\vert_{KW} \leq \alpha \bigr\}$. 

As an immediate consequence of the above sections, it follows that some of the remarks already mentioned for the case of the Kechris-Woodin ranking still apply in the case of the solvable ranking:

\begin{proposition}
Let $y:[0,1] \to [0,1]$ be a differentiable function. Then $\vert y \vert_{SV}= 1$ if and only if 
$y$ is continuously differentiable.
\end{proposition}

\begin{proposition}
Let $y:[0,1] \to [0,1]$ be the function defined as $y(0)=0$ and $y(x)= x^2\sin(1/x)$ otherwise. Then $\vert y \vert_{SV}= 2$.
\end{proposition} 

Nonetheless, the apparent similarities do not translate to an equivalence in terms of ranking for the two different notions. Indeed, we can prove the following theorem: 

\begin{theorem}
\label{thm:diffranking}
We have $KW^{*}_{1}=SV_{1}$. Then for all $1 <\alpha< \omega_1$ we have $ KW^{*}_{\alpha} \subseteq  SV_{\alpha}$. Moreover, for all $1 <\alpha< \omega_1$, we have $ KW^{*}_{\alpha+1} \setminus KW^{*}_{\alpha} \neq  SV_{\alpha+1} \setminus SV_{\alpha}$. 
\end{theorem}

\begin{proof}
The case of $KW^{*}_{1}=SV_{1}$ is trivial and follows immediately from the above remarks. Indeed, $SV_{1} = KW^{*}_{1}= KW_{1}=C^1([0,1])$ is the set of continuously differentiable functions. To continue, we recall from \cite{kechris1986ranks} that making use of rational numbers or real numbers in the definition of the derivative operator in the Kechris-Woodin ranking produces the same ranking. More formally, let $y \in C([0,1])$, $\epsilon>0$ and $P \subset [0,1]$ closed, if we define the operator as: 
 
\begin{equation*}
\begin{split}
P^{\star}_{\epsilon,y}= \{ x \in P : \forall \delta >0 \exists p < q , r < s \in B_{[0,1]}(x, \delta) \; \text{with} \\
 [p,q] \cap [r,s] \cap P \neq \emptyset \; \text{and} \; \vert \Delta_y (p,q) - \Delta_y (r,s) \vert \geq \epsilon \}
\end{split}
\end{equation*} 

we obtain, for all $\epsilon >0$, the following chain of inequalities $P^{\star}_{\epsilon,y} \subseteq P'_{\epsilon/2,y} \subseteq P^{\star}_{\epsilon/2,y}$. Hence, using rationals or reals in the definition of the involved operator produces the same ranking. Therefore, we will make use of the star derivative operator in this proof. Moreover, for all $y \in SV$ let us use the notation $\{ E_\gamma \}_\gamma$ for the sequence of $y'$-removed sets on $[0,1]$. 

We prove the statement of the theorem by proving separately the two following conditions: 

\begin{enumerate}
\item For all $y \in SV$, for all $\alpha< \omega_1$,  having $\left\vert y \right\vert_{SV}= \alpha$ implies $\left\vert y \right\vert_{KW} \geq \alpha$ \\ 
\item For all $1< \alpha< \omega_1$, $\exists y \in SV$ such that $\left\vert y \right\vert_{KW} = \alpha+1$ and $\left\vert y \right\vert_{SV} = 2$ 
\end{enumerate}

We prove the above conditions separately. 

\begin{enumerate}
\item For all $y \in SV$, for all $\alpha< \omega_1$,  having $\left\vert y \right\vert_{SV}= \alpha$ implies $\left\vert y \right\vert_{KW} \geq \alpha$ \\ 

Note that this statement comes for free if we show that, for all $\alpha< \omega_1$, for all $y \in SV$, we have $E_\alpha \subseteq \bigcup_\epsilon P_{\epsilon,y}^{\alpha}$. Indeed, if $\left\vert y \right\vert_{SV}= \alpha$, that means that $E_\alpha = \emptyset$, and therefore, since $E_\alpha \subseteq \bigcup_\epsilon P_{\epsilon,y}^{\alpha}$, we can either have $\bigcup_\epsilon P_{\epsilon,y}^{\alpha}= \emptyset$, in which case $\left\vert y \right\vert_{KW}= \left\vert y \right\vert_{SV} = \alpha$, or we can have $\bigcup_\epsilon P_{\epsilon,y}^{\alpha} \neq \emptyset$, which means $\left\vert y \right\vert_{KW}> \left\vert y \right\vert_{SV}= \alpha$. But if $E_\alpha \neq \emptyset$, which means $\left\vert y \right\vert_{SV} > \alpha$, then since $E_\alpha \subseteq \bigcup_\epsilon P_{\epsilon,y}^{\alpha}$, that implies $\bigcup_\epsilon P_{\epsilon,y}^{\alpha} \neq \emptyset$, which in turn means that $\left\vert y \right\vert_{KW}> \alpha$. We now prove  by transfinite induction that for all $\alpha< \omega_1$, for all $y \in SV$, we have $E_\alpha \subseteq \bigcup_\epsilon P_{\epsilon,y}^{\alpha}$. 

\begin{itemize}
\item {\bf{Base of the induction}}: for all $y \in SV$ we prove that, for all $x \in [0,1]$, if $x \in E_1$ then $x \in \bigcup_\epsilon P_{\epsilon,y}^{1}$. 

Recall that by definition $E_1$ corresponds with the set of discontinuity points of $y'$ on $[0,1]$. Hence, the desired result is directly proved in \cite{kechris1986ranks}[Fact 3.3] \\

\item {\bf{Successor ordinals}}: for all $\alpha< \omega_1$, for all $y \in SV$, we prove that, if $E_\alpha \subseteq \bigcup_\epsilon P_{\epsilon,y}^{\alpha}$, then $E_{\alpha+1} \subseteq \bigcup_\epsilon (P_{\epsilon,y}^{\alpha})^{\star}$. 

For all $y \in SV$, since for all $\alpha< \omega_1$ we know that $E_{\alpha+1} \subset E_{\alpha}$ the induction hypothesis implies $E_{\alpha+1} \subseteq \bigcup_\epsilon P_{\epsilon,y}^{\alpha}$. Moreover, recall that:

\begin{equation*}
 E_{\alpha+1}= \{ x \in E_{\alpha} : y' \restriction_{E_{\alpha}} \; \text{is discontinuous on } \; x \; \}
\end{equation*}

For all points $x \in  E_{\alpha+1}$, since $y' \restriction_{E_{\alpha}}$ is discontinuous on $x$, it means that $x$ is an accumulation point for $E_{\alpha} \subseteq \bigcup_\epsilon P_{\epsilon,y}^{\alpha}$. Therefore, by definition of a discontinuity point, this means that: 

\begin{equation*}
\begin{split}
\forall \delta >0 \; \exists \epsilon>0 \; \text{ such that } \; \exists p(\epsilon) < q(\epsilon) , r(\epsilon) < s(\epsilon) \in B_{E_{\alpha}}(x, \delta) \\
 \text{with } \; [p(\epsilon),q(\epsilon)] \cap [r(\epsilon),s(\epsilon)] \cap x \neq \emptyset \; \text{and} \; \left\vert \Delta_y (p(\epsilon),q(\epsilon)) - \Delta_y (r(\epsilon),s(\epsilon)) \right\vert \geq \epsilon 
\end{split}
\end{equation*}

Hence, that means that: 

\begin{equation*}
\begin{split}
x \in \bigcup_\epsilon \{ z \in  P_{\epsilon,y}^{\alpha} : \forall \delta >0 \exists p < q , r < s \in B_{[0,1]}(z, \delta) \; \text{with} \\
 [p,q] \cap [r,s] \cap  P_{\epsilon,y}^{\alpha} \neq \emptyset \; \text{and} \; \vert \Delta_y (p,q) - \Delta_y (r,s) \vert \geq \epsilon \} \\
 = \bigcup_\epsilon (P_{\epsilon,y}^{\alpha})^{\star}= \bigcup_\epsilon P_{\epsilon,y}^{\alpha+1}
\end{split}
\end{equation*} 

\item {\bf{Limit ordinals}}: for all $y \in SV$ we prove that, if $E_\alpha \subseteq \bigcup_\epsilon P_{\epsilon,y}^{\alpha}$ for all $\alpha < \lambda$, then $E_{\lambda} \subseteq \bigcup_\epsilon P_{\epsilon,y}^{\lambda}$. 

Once again, for all $y \in SV$ the induction hypothesis implies $E_{\lambda} \subseteq \bigcup_\epsilon P_{\epsilon,y}^{\alpha}$ for all $\alpha < \lambda$. Note that this, combined with the previous case of the induction for successor ordinals, implies also that $E_{\lambda} \subseteq \bigcup_\epsilon P_{\epsilon,y}^{\alpha+1} = \bigcup_\epsilon ( P_{\epsilon,y}^{\alpha})^{\star}$ for all $\alpha < \lambda$. But since by definition $P_{\epsilon,y}^{\lambda}= \bigcap_{\alpha< \lambda} (P^{\alpha}_{y,\epsilon})^{\star}_{y,\epsilon}$ it follows that $E_{\lambda} \subseteq \bigcup_\epsilon P_{\epsilon,y}^{\lambda}$. 

\end{itemize}

We have therefore completed the proof of 1. \\

\item For all $1< \alpha< \omega_1$, $\exists y \in SV$ such that $\left\vert y \right\vert_{KW} = \alpha+1$ and $\left\vert y \right\vert_{SV} = 2$ \\

Note that proving the above statement implies that for all $1 <\alpha< \omega_1$ it is true that $ KW^{*}_{\alpha+1} \setminus KW^{*}_{\alpha} \neq  SV_{\alpha+1} \setminus SV_{\alpha}$. 

We make use of the techinque developed in \cite{westrick2014lightface} in order to define, for all $\alpha < \omega_1$, a differentiable function $y$ such that $\left\vert y \right\vert_{KW}= \alpha+1$. The definition is based on the idea of associating to each well-founded tree $T$, a specific function $g_T$ whose derivative has a complexity ranking connected to the limsup rank of the tree $T$. More precisely, this is done in the following manner: let $r:[0,1] \rightarrow \mathbb{R}$ be a computable function satisfying: 

\begin{itemize}
\item $r$ is continuously differentiable 
\item $r\left(\frac{1}{2}\right)=\frac{1}{2}$ 
\item $r(0)=r(1)=r^{\prime}(0)=r^{\prime}(1)=0$ 
\item $\|r\|<1$ and $\left\|r^{\prime}\right\|<2$
\end{itemize}

Then let $\left\{\left[a_n, b_n\right]\right\}_{n \in \mathbb{N}_0}$ be any computable sequence of intervals with rational endpoints satisfying: 

\begin{itemize}
\item Each interval is contained in $\left(\frac{1}{4}, \frac{1}{2}\right)$ 
\item $b_{n+1}<a_n<b_n$ for each $n$. 
\item $\lim _{n \rightarrow \infty} a_n=\frac{1}{4}$ 
\item $b_n-a_n<\left(a_n-\frac{1}{4}\right)^2$ for each $n$
\end{itemize}

Then for any well-founded tree $T \in WF$, define function $g_T: [0,1] \to \mathbb{R}$ as follows: 

\begin{equation}
g_T(x)= 
\begin{cases}
0 \hspace{0.7 cm} \text{  if  } T= \emptyset \\
r(x) + \sum_{n=0}^{\infty} g_{T_{n}}[a_n,b_n] \hspace{0.7 cm} \text{otherwise}
\end{cases}
\end{equation}

Where as usual for all $n \in \mathbb{N}_0$ the notation $T_n$ denotes the nth subtree of $T$, i.e.\ $T_n= \{ \sigma : n^\frown \sigma \in T \}$ and $g_{T_{n}}[a_n,b_n]$ denotes the function which is identically 0 outside of $I_n$, and instead values $g_{T_{n}}[a_n,b_n](x)=(b_n-a_n) g_{T_{n}} \left(\frac{x-a_n}{b_n-a_n}\right)$ for all $x \in [a_n,b_n]$.

It has been proved in \cite{westrick2014lightface} that for all $\alpha < \omega_1$ there exists a tree $T(\alpha)$ with $\left\vert T(\alpha) \right\vert_{ls}=\alpha+1$ such that $g_{T(\alpha)} \in \mathcal{D}$ and $\left\vert g_{T(\alpha)} \right\vert_{KW}= \alpha+1$. It is then clear from the above discussions that if $1<\alpha < \omega_1$ then such $g'_{T(\alpha)}$ is discontinuous on $[0,1]$; indeed, this follows directly from the base case of the transfinite induction above. We want now to show that $g'_{T(\alpha)}$ is bounded and solvable. For all $1<\alpha < \omega_1$ it is not hard to see that $g'_{T(\alpha)}$ is bounded by induction on the rank of the tree; indeed if every $g'_{T_n (\alpha)}$ is bounded then $g'_{T(\alpha)}$ is as well, being the sum of bounded functions that have disjoint support. Indicating the set of discontinuity points of $g'_{T(\alpha)}$ as $D_{g'_{T(\alpha)}}$, we show $D_{g'_{T(\alpha)}}$ is a closed set and that $g'_{T(\alpha)} \restriction_{D_{g'_{T(\alpha)}}}$ is the constant function identically equal to zero. First, it follows from the definition of $g_{T(\alpha)}$ above that the discontinuity points for $g'_{T(\alpha)}$ are of the form $ a_n + \frac{b_n- a_n}{4}$ for some rational interval of the form $[a_n,b_n]$ where $n \in \mathbb{N}_0$. Moreover, each of these discontinuity points $p\in D_{g'_{T(\alpha)}}$ is the limit point for at least one sequence of such intervals. If only a finite number of such intervals in the sequence contains a discontinuity point, then clearly $p$ is an isolated point of $D_{g'_{T(\alpha)}}$ and therefore belongs to its closure, i.e.\ $p \in \operatorname{cl} ( D_{g'_{T(\alpha)}})$. Otherwise, there is a countable sequence of points in $D_{g'_{T(\alpha)}}$ converging to $p$. In any case, for all $p \in D_{g'_{T(\alpha)}}$ we have $p \in \operatorname{cl} ( D_{g'_{T(\alpha)}})$, which means that $ D_{g'_{T(\alpha)}}$ is a closed set. Finally, we show that $g'_{T(\alpha)} \restriction_{D_{g'_{T(\alpha)}}}$ is the constant function identically equal to zero. From the definition above it follows that on any neighborhood of any discontinuity point the behavior of $g_{T(\alpha)}$ is bounded quadratically, i.e.\ for all points $p \in D_{g'_{T(\alpha)}}$, for small enough $\epsilon$ we have $\left\vert g_{T(\alpha)}(x) \right\vert \leq (x-p)^2$ for all $x \in (p-\epsilon,p+\epsilon)$. This implies that for all points $p \in D_{g'_{T(\alpha)}}$ we have $g'_{T(\alpha)}(p) = \lim_{\epsilon \to 0}  \frac{\left\vert g_{T(\alpha)}(p+\epsilon)- g_{T(\alpha)}(p) \right\vert}{\epsilon} \leq \lim_{\epsilon \to 0} ( \epsilon^2/\epsilon = \epsilon )=0$. But then, since $g'_{T(\alpha)} \restriction_{D_{g'_{T(\alpha)}}}$ is the constant function identically equal to zero, it means that $E_2= \emptyset$ and so $\left\vert g_{T(\alpha)} \right\vert_{SV}= 2$. This concludes the proof of 2. and therefore concludes the proof of the theorem. 

\end{enumerate}
\end{proof}



The existing relation between solvable and Kechris-Woodin rankings as described in the theorem above allows to devise further relations with other existing rankings for differentiable functions such as the Denjoy ranking and the Zalcwasser ranking. Since the exact definitions of such rankings are not essential for the scope of this paper, we limit ourselves to providing intuitive descriptions while discussing such relations. 

The Denjoy ranking is a ranking that is focused on the number of steps required for the Denjoy totalization process to converge. In this sense, it has a close similarity with our solvable ranking, which is also based on the termination of a similar process based on transfinite recursion. Nonetheless, its focus is on the integration of the derivative, while ours is on the convergence of our search method within the domain. This generates relevant differences among the two rankings as we now show. Moreover, the Denjoy ranking can be considered from two different angles or perspectives, depending on which base space is chosen for the description. As an example, in \cite{kechris1986ranks} and \cite{ki1997denjoy}, the Denjoy ranking is defined over the set of differentiable functions $\mathcal{D}$, and the ranking is attributed to a function $y \in \mathcal{D}$ according to how many steps of the Denjoy process are required from $y'$ to obtain $y$. In \cite{Wes20} instead, the ranking is directly attributed to the derivatives $y'$, considered as elements of the Polish space of measurable functions. According to this description, it is proved there \cite{Wes20}[Theorem 22] that the set of Denjoy integrable functions is $\boldsymbol{\Pi}_1^1$ complete in the space of measurable functions. This description is more general since it includes not only derivatives but also functions that are derivatives almost everywhere. In any case, it is clear that, since the base spaces considered are different for the two approaches, the rankings are also not directly comparable. For our analysis and comparison with solvable functions, we choose to look at the Denjoy ranking from the point of view of \cite{kechris1986ranks} and \cite{ki1997denjoy}, since this attitude aligns with the spirit of our definition. In this context, a theorem from \cite{Ra91} proves that for every differentiable function, its Kechris-Woodin ranking is greater or equal to its Denjoy ranking. Then, thanks to our theorem above, this implies that for every function in $SV$, its solvable ranking is also greater or equal to its Denjoy ranking. To be precise, this is just an immediate consequence of a trivial observation connected with Lebesgue integration. Indeed, the set of solvable functions $SV$ contains only differentiable functions with bounded derivatives, which are all by definition Lebesgue integrable functions. Hence, for all functions $y \in SV$ we have $\left\vert y \right\vert_{D}=1$, where $\left\vert y \right\vert_{D}$ stands for the Denjoy ranking of $y$ according to the spirit of \cite{kechris1986ranks}. The latter also implies that the Denjoy ranking as intended this way can not be a $\boldsymbol{\Pi}_1^1$ norm on $\mathcal{D}$, since that would mean that the set of all functions of Denjoy ranking equal 1 should be Borel while at the same time include all bounded derivatives, whose Kechris-Woodin rankings are unbounded below $\omega_1$. We will instead show that the solvable ranking is indeed a $\boldsymbol{\Pi}_1^1$ norm on $SV$. 

The Zalcwasser ranking shares a feature with the solvable ranking: it is not strictly a ranking for differentiable functions. But while the solvable ranking on one dimension introduces a ranking for a proper subset of differentiable functions, the Zalcwasser ranking introduces a ranking for a set that contains all differentiable functions. Indeed, in this context, the base space is not anymore $C([0,1])$ but instead $EC$, i.e.\ the set of continuous functions with everywhere convergent Fourier series. It is known \cite{AK87} that $EC$ is a $\boldsymbol{\Pi}_1^1$ complete subset of $C(\mathbb{T})$, where $\mathbb{T}$ represents the unit circle, i.e.\ $[0, 2\pi]$ with $0$ identified with $2\pi$. To provide an intuitive description, the Zalcwasser ranking of a function $y \in EC$, which is indicated with $\left\vert y \right\vert_{Z}$, rates how uniform is the convergence of the Fourier serie of $y$. It is demonstrated in \cite{AK87} that the Zalcwasser ranking is a $\boldsymbol{\Pi}_1^1$ norm on $EC$. If we indicate with $\mathcal{D}(\mathbb{T})$ the set of differentiable functions over the unit circle, then we have $\mathcal{D}(\mathbb{T}) \subseteq EC$. Hence, once restricting the analysis to the unit circle instead of the unit interval, it is possible to compare the Zalcwasser ranking with the other rankings for differentiable functions. While the Denjoy ranking and the Zalcwasser ranking are proved to be incomparable, \cite{Ki95} illustrates that for every differentiable function, its Kechris-Woodin ranking is greater or equal to its Zalcwasser ranking. If we call $SV(\mathbb{T})$ the set of solvable functions defined exactly as before but over the unit circle $\mathbb{T}$, then our theorem above yields the following corollary:

\begin{corollary}
For every solvable function $y \in SV(\mathbb{T})$ we have $\left\vert y \right\vert_{SV} \geq \left\vert y \right\vert_{Z}$
\end{corollary}

Note that the statement of the above corollary is not as trivial as its analog for the case of the Denjoy ranking, since it is known that for every countable ordinal $\alpha < \omega_1$ there exists a function in $\mathcal{D}(\mathbb{T})$ satisfying $ \left\vert y \right\vert_{Z}= \alpha$. Therefore, Theorem \ref{thm:diffranking} is needed to prove the above corollary. 

\subsection{Unboundness below $\omega_1$}

We now introduce a Borel encoding from well-founded trees to solvable functions that reflects the solvable ranking, i.e.\ a Borel function associating to each tree $T \in WF$ that satisfies $\left\vert T \right\vert_{ls}= \alpha$ for some $\alpha< \omega_1$ a differentiable function $y_T$ such that $\left\vert y_T \right\vert_{SV} = \alpha+1$. 


We start by providing the following definition: 

\begin{definition}
\label{def:treefunction}
Let $\{ I_n \}_{n \in \mathbb{N}_0}$ be an enumeration of the removed open intervals from the construction of the Cantor set $C$. Call $p: [0,1] \to \mathbb{R}$ the function in $SV$ defined by cases as:

\begin{equation}
\begin{cases}
p(x)= x^2 \sin \left( \frac{1}{x} \right) \hspace{0.7 cm} \text{  if  } x \in (0, \bar{x}] \\
p(x)= \bar{x}^2 \sin \left( \frac{1}{\bar{x}} \right) \hspace{0.7 cm} \text{  if  } x \in [\bar{x}, 1-\bar{x}] \\
p(x)= -(x-1)^2 \sin \left( \frac{1}{x-1} \right) \hspace{0.7 cm} \text{  if  } x \in [1-\bar{x}, 1) \\
p(x)= 0 \hspace{0.7 cm} \text{ otherwise}
\end{cases}
\end{equation}

where the point $\bar{x}$ is the greatest point in $[0,1]$ such that $\bar{x} < \frac{1}{2}$ and that $ 2\bar{x} \sin{ \left( \frac{1}{\bar{x}}\right) } - \cos{ \left( \frac{1}{\bar{x}} \right)}=0$. \\

For all $T \in WF$ define $y_T: [0,1] \to \mathbb{R}$ as the following function:
\begin{equation}
y_T(x)= 
\begin{cases}
0 \hspace{0.7 cm} \text{  if  } T= \emptyset \\
p(x) + \frac{1}{4} \sum_{n=0}^{\infty} y_{T_n}[I_n] \hspace{0.7 cm} \text{otherwise}
\end{cases}
\end{equation}
\end{definition}

The function defined this way satisfies some nice properties expressed by the following proposition: 

\begin{proposition}
For all $T \in WF$, let $y_T$ be the function defined above. Then: 
\begin{enumerate}
\item If $T$ is a computable tree, $y_T$ is computable in $T$
\item $y_T(0)=y_T(1)=y'_T(0)=y'_T(1)=0$
\item $y_T$ is differentiable and we have $\left\Vert y_T \right\Vert<2$ and $\left\Vert y'_T \right\Vert<2$.
\item $y_T$ is solvable 
\end{enumerate} 
\end{proposition}

\begin{proof}
If $T$ is empty the statement is trivial. We therefore prove it for all nonempty trees $T \in WF$. We prove these properties by induction on the rank of the tree. For the base case, it is easy to check that all the properties are satisfied due to the behavior of function $p$ which is very similar to the function already analyzed for Example \ref{ex:firstex}. Assume now that all the properties hold for all trees with rank less than $\left\vert T \right\vert$. Then on any interval whose closure does not contain a point in the Cantor set $C$ the function $y_T$ is equal to a finite sum of computable functions and so it is computable. The fact that the sum is finite follows from the fact that $T$ has no infinite branches being well-founded. Let us now analyze computability for the points in $C$. For all $c \in C$, for small enough $\epsilon$, consider the open interval $(c-\epsilon, c+\epsilon)$; for all $x \in (c-\epsilon, c+\epsilon)$ we have that either $x \in C$ or $x \in I_n=(a_n,b_n)$ for some $n \in \mathbb{N}_0$. This is because all removed intervals from the construction of $C$ are disjoint. In the first case, we have $y_T(x)= p(x)$ while in the second case instead $y_T(x)= p(x) + \frac{1}{4} y_{T_n}[I_n](x)$, since all the other terms in the sum equal to zero. Since the size of the removed intervals is decreasing exponentially it is always possible to choose $\epsilon$ small enough such that for all $I_n \cap  (c-\epsilon, c+\epsilon)$ with $b_n, a_n \neq c$ we have $b_n - a_n < \epsilon$. Therefore, for all $x \in (c-\epsilon, c+\epsilon)$ we have $\left\vert y_T(x) - p(x)\right\vert <  \left\Vert y_{T_n} \right\Vert \epsilon = 2 \epsilon$ which means, since $p$ is computable, that $y_T$ is computable.

The fact that $y_T(0)=y_T(1)=y'_T(0)=y'_T(1)=0$ is trivially satisfied by definition. For differentiability, we get in the same fashion that if $y_{T_n}(0)=y_{T_n}(1)=y'_{T_n}(0)=y'_{T_n}(1)=0$ and each $y_{T_n}$ is differentiable for all $n \in \mathbb{N}_0$ then also each $y_{{T}_n}[I_n]$ is differentiable. Then once again $y_T$ is surely differentiable on $[0,1] \setminus C$ since there it is equal to a finite sum of differentiable functions. On the other hand, we have just shown that for all $c \in C$, for small enough $\epsilon$ and all $x \in (c-\epsilon, c+\epsilon)$ we have $\left\vert y_T(x) - p(x)\right\vert <  2 \epsilon$ which implies that $y_T$ is differentiable also on $C$ and that for all $c \in C$ we have $y'_T(c)=p'(c)$.

We now prove that $\left\Vert y_T \right\Vert<2$ and $\left\Vert y'_T \right\Vert<2$ by induction. Since all the removed intervals are disjoint, and their length is always bounded by $\frac{1}{3}$, for all $x \in [0,1] \setminus C$ we have: 

\begin{equation*}
 \left\vert y_T(x) \right\vert \leq \left\vert p(x) \right\vert + \frac{1}{12} \left\Vert y_{T_n} \right\Vert \leq 1+\frac{1}{6}<2
\end{equation*}

and:

\begin{equation*}
 \left\vert y'_T(x) \right\vert \leq \left\vert p'(x) \right\vert + \frac{1}{4} \left\Vert y'_{T_n} \right\Vert < \frac{3}{2} +\frac{1}{2}=2
\end{equation*}

Instead, since for all $c \in C$, for small enough $\epsilon$ and for all $x \in (c-\epsilon, c+\epsilon)$ we have $\left\vert y_T(x) - p(x)\right\vert <  2 \epsilon$ that implies that $y_T(c)=p(c)$ and so $ \left\vert y_T(c) \right\vert <2$. In the same way, since for all $c \in C$ we have $y'_T(c)=p'(c)$ it follows $\left\vert y'_T(c) \right\vert <2$.

Finally, it is easy to see that $y_T$ is solvable; indeed once again the statement is true for all points in $[0,1] \setminus C$ by induction hypothesis since each of the $y_{T_n}$ for all $n \in \mathbb{N}_0$ is solvable. On $C$, which is a closed set, it is clear that $y'_T$ is discontinuous. Moreover, since $y'_{T} \restriction_C=p'\restriction_C$, that means that for each closed subset $K$ of $C$ the set of discontinuity points of $y'_{T} \restriction_K$ is also a closed subset of $K$. Indeed, $y'_{T} \restriction_C=p'\restriction_C$ is discontinuous only on $\{ 0, 1\}$. 

\end{proof}

This last proposition sets the premises for the next proposition, which establishes the possibility of connecting the limsup rank for $WF$ to the solvable ranking for solvable derivatives. But since the definition of the function $y_T$ is quite involved and hard to visualize at first sight,  before proceeding, we provide some examples of the simplest cases to help the intuition. 

Let us start by considering a well-founded tree $T \in WF$. If this tree is empty, then $y_T$ is just the constant function with value $0$, which is trivially continuously differentiable. Hence, $\left\vert y_T \right\vert_{SV}= 1 = \left\vert T \right\vert_{ls}+1$. Instead, let us assume $T$ is just a root, which means that for all $n$ we have $T_n= \emptyset$; and so $\sup_n \left\vert T_n \right\vert_{ls} = \limsup_n \left\vert T_n \right\vert_{ls} = 0$, from which it follows $\left\vert T \right\vert_{ls} =1$. In such case then $y_{T}= p$. It is easy to see that its derivative values $0$ in the points $0$ and $1$ and it is discontinuous in such points. Because of one previous example, we know that this function is solvable and that its solvable rank is $2$. Hence, also in this case, $\left\vert y_T \right\vert_{SV}= \left\vert T \right\vert_{ls}+1$. Let us now assume that the tree $T$ has a root and one child. This means that the related function has the form $y_{T}(x)= p(x) + \frac{1}{4} p[1/3,2/3](x)$ for all $x \in [0,1]$. This function can be described by cases, but more than the exact expression, what matters for the intuition is that it is a differentiable function and its derivative is bounded and discontinuous on the points $0$, $1/3$, $2/3$, $1$. Moreover, we have $y'_{T}(0)=y'_{T}(1)=0$, $y'_{T}(1/3)=p'(1/3)$ and $y'_{T}(2/3)=p'(2/3)$. Hence, also in this case we have $\left\vert y_{T} \right\vert_{SV}=2$ and $\left\vert y_T \right\vert_{SV}= \left\vert T \right\vert_{ls}+1$. An interesting case is the case of a tree $T$ with one root with infinite children. In this case, if we indicate with $(a_n, b_n)$ the $n$th removed interval in the enumeration for the Cantor set, then we have $y_{T}(x)= p(x) + \frac{1}{4} \sum_{n=0}^{\infty} p[a_n,b_n](x)$ for all $x \in [0,1]$, which once again can be expressed by cases. It is not hard to see that this function is differentiable and that its derivative is bounded. Moreover, $y'_T$ is discontinuous in all points of the Cantor set $C$ and we have $y'_T(c) = p'(c)$ for all $c \in C$. It then follows that $y'_T \restriction_C$ is discontinuous in $0$ and $1$ where it values $0$. Therefore, this means that in this case $\left\vert y_{T} \right\vert_{SV}=3= \left\vert T \right\vert_{ls}+1$. 

Intuitively, this brief description of the behavior of $y_T$ for these simple trees of low rank already outlines one of the key differences between this function and the function $g_T$ as described in the proof of Theorem \ref{thm:diffranking} and as originally defined in \cite{westrick2014lightface}. Both definitions have the goal of associating to each well-founded tree a specific differentiable function that reflects the limsup ranking of the tree. Nonetheless, $g_T$ uniformly achieves this goal by always forcing its derivative $g'_T$ to value zero in each of its discontinuity points, no matter how \emph{complex} the derivative is in these points according to the Kechris-Woodin notion of complexity. Instead, the design of $y_T$ allows its derivative $y'_T$ to automatically keep track of such complexity in each of these points just by storing such information in the values that $y'_T$ assumes around them. In other words, the additional information that $y'_T$ carries along is not just describing the Kechris-Woodin ranking of $y_T$ but also its solvable ranking. Recall that thanks to Theorem \ref{thm:diffranking} we know that the solvable ranking of any function in $SV$ is always greater or equal to its Kechris-Woodin ranking. 

Formally, the only last element that is left to prove concerning function $y_T$ is indeed the correct association of the limsup ranking with the solvable ranking for any countable ordinal. This is expressed by a following proposition. Nonetheless, we first state an observation that is also relevant for the proof of such a proposition. 

\begin{remark}
Let $y \in SV$ be a function such that $y(0)=y(1)=y'(0)=y'(1)=0$. For all $[a,b] \subseteq [0,1]$ we have $ \left\vert \frac{1}{4} y[a,b] \right\vert_{SV} = \left\vert y \right\vert_{SV}$
\end{remark}

\begin{proof}
To see that this remark is true, recall that the solvable ranking of $y$ is concerned with the discontinuities of the derivative $y'$ and its sequence of $y'$-removed sets on $[0,1]$. Recall also that, by definition, $\frac{1}{4} y[a,b](x)=0$ if $x \in [0,1] \setminus [a, b]$, and for $x \in[a, b]$ we have $\frac{1}{4}y[a, b](x)= \frac{1}{4} (b-a) y\left(\frac{x-a}{b-a}\right)$. Hence, since $y(0)=y(1)=y'(0)=y'(1)=0$ by hyphotesis, the function $\frac{1}{4} y[a,b]$ is differentiable. Moreover, $\frac{1}{4}y'[a,b](x)=0$ if $x \in [0,1] \setminus [a, b]$ and for $x \in[a, b]$ we have $(\frac{1}{4} y[a, b])'(x)= \frac{1}{4} y'\left(\frac{x-a}{b-a}\right)$. Therefore, each point of discontinuity $x$ of $y'$ on $[0,1]$ has a biunivocal correspondence with a point of discontinuity $\frac{x-a}{b-a}$ of $(\frac{1}{4} y[a, b])'$ on $[a,b]$. This fact easily extends to further levels of the sequence of $y'$-removed sets on $[0,1]$ which correspond in the same way to levels of the sequence of $(\frac{1}{4}y[a, b])'$-removed sets on $[0,1]$ yielding the same solvable ranking for the two functions. 
\end{proof}

We can now present the proposition mentioned above. 

\begin{proposition}
\label{prop:unboundedsolv}
For all $T \in WF$ with $\left\vert T \right\vert_{ls}=\alpha$ the function $y_T$ defined as above satisfies $ \left\vert y_T \right\vert_{SV}= \alpha +1$.
\end{proposition}

\begin{proof}
We prove the proposition by transfinite induction. The base case has already been proved in the above discussion. As induction hypothesis let us now fix some $1< \alpha < \omega_1$ and assume that, for all $T \in WF$ with $\left\vert T \right\vert_{ls}= \gamma$ and $\gamma< \alpha$, we have $\left\vert y_T \right\vert_{SV}=\gamma+1$. 

Let us now consider, for all $T \in WF$ with $\left\vert T \right\vert_{ls}=\alpha$, the expression of function $y_T$, which is, by definition $y_T(x)= p(x) + \frac{1}{4} \sum_{n=0}^{\infty} y_{T_n}[I_n](x)$ for all $x \in [0,1]$. Our goal is to prove that $\left\vert y_T \right\vert_{SV}= \alpha +1$ which means, by definition of the solvable ranking, proving that $E_{\alpha} \neq \emptyset$ and $E_{\alpha+1}= \emptyset$, where as usual the notation $\{ E_{\gamma} \}_{\gamma<\omega_1}$ stands for the sequence of $y'_T$-removed sets on $[0,1]$. Note now that by definition of the limsup rank over well-founded trees, which is always a successor ordinal, letting $\left\vert T \right\vert_{ls}=\alpha=\beta+1$ implies two possible scenarios:

\begin{enumerate}
\item $\sup_n \left\vert T_n \right\vert_{ls}< \alpha$ and $\limsup_n \left\vert T_n \right\vert_{ls}= \beta$
\item  $\sup_n \left\vert T_n \right\vert_{ls}= \alpha$
\end{enumerate}

So let us analyze each of these scenarios separately: 

\begin{enumerate}

\item $\sup_n \left\vert T_n \right\vert_{ls}< \alpha$ and $\limsup_n \left\vert T_n \right\vert_{ls}= \beta$ \\ 

In this first case, it is necessary to make a distinction on whether $\beta$ is a successor ordinal or a limit ordinal. This once again leads to a distinction by cases: 

\begin{enumerate}
\item $\beta$ is a successor ordinal: \\

We then infer that there exists a $n_\beta \in \mathbb{N}_0$ such that all the subtrees $T_i$ for all $i \geq n_\beta$ satisfy $\left\vert T_i \right\vert_{ls}= \beta$ while there does not exist a $j \in \mathbb{N}_0$ such that $\left\vert T_j \right\vert_{ls}= \alpha$.  By induction hyphotesis we know that for all the subtrees $T_i$ for all $i \in \mathbb{N}_0$ such that $\left\vert T_i \right\vert_{ls}= \gamma$ with $\gamma< \alpha$ we have $\left\vert y_{T_i} \right\vert_{SV}= \gamma +1$. Therefore, this means that for all the subtrees $T_i$ for all $i \geq n_\beta$ we have $\left\vert y_{T_i} \right\vert_{SV}= \alpha$. Hence, due to the remark above, we have that $\left\vert \frac{1}{4} y_{T_i}[I_i] \right\vert_{SV}= \alpha$ for all $i \geq n_\beta$. Since $y'_T= p' + \frac{1}{4} \sum_{n=0}^{\infty} y'_{T_n}[I_n]$ and $p'$ is continuous on all $\operatorname{cl}(I_n)$ for all $n \in \mathbb{N}_0$, this means that for all $i \geq n_\beta$ we have $\operatorname{cl}(I_i) \cap E_\beta \neq \emptyset$ and $\operatorname{cl}(I_i) \cap E_\alpha = \emptyset$. 

Therefore, each point of the Cantor set $C$ is either in $E_\beta$ or is an accumulation point for $E_\beta$ and, since $E_\beta$ is a closed set, this means $C \subseteq E_\beta$. Let us now analyze the behavior of  $y'_T \restriction_{E_\beta}$. We observe that $y'_T \restriction_{E_\beta}$ is continuous on all $I_n$ for all $n \in \mathbb{N}_0$. Indeed, we have just shown that for all $i \geq n_\beta$ we have $I_i \cap E_\alpha = \emptyset$ and so that implies continuity of $y'_T \restriction_{E_\beta}$ on such intervals. Instead, for all $n < n_\beta$, we know by hyphosesis that $\left\vert y_{T_n} \right\vert_{SV} < \alpha$ which implies $I_n \cap E_\beta = \emptyset$. We are now left to consider the behavior of $y'_T \restriction_{E_\beta}$ on the Cantor set $C$. It is easy to see that $y'_T \restriction_{C}=p'\restriction_{C}$. Indeed each other term in the sum equals zero on the Cantor set by the definition of each function $y'_{T_n}[I_n]$ for all $n \in \mathbb{N}_0$. But since $p'\restriction_{C}$ is discontinuous on the set $\{ 0,1 \}$, this means that $y'_T \restriction_{E_\beta}$ is discontinuous on the set $\{ 0,1 \}$ and so $E_\alpha =\{ 0,1 \}$. Note that by definition we know that $y'_T \restriction_{E_\alpha}(0)=y'_T \restriction_{E_\alpha}(1)=0$. It follows $E_{\alpha+1} = \emptyset$ and so $\left\vert y_T \right\vert_{SV}= \alpha +1$. \\

\item $\beta$ is a limit ordinal: \\

We then infer that for all ordinals $\gamma < \beta$ there exists an ordinal $\gamma<\delta< \beta$ such that $T_{i_\delta}$ satisfy $\left\vert T_{i_\delta} \right\vert_{ls}= \delta$ while there does not exist a $j \in \mathbb{N}_0$ such that $\left\vert T_j \right\vert_{ls}= \alpha$.  By induction hyphotesis we know that for all the subtrees $T_i$ for all $i \in \mathbb{N}_0$ such that $\left\vert T_i \right\vert_{ls}= \gamma$ with $\gamma< \alpha$ we have $\left\vert y_{T_i} \right\vert_{SV}= \gamma +1$. Therefore, this means that for all ordinals $\gamma < \beta$ there exists an ordinal $\gamma<\delta< \beta$ such that $T_{i_{\delta}}$ satisfy $\left\vert y_{T_{i_{\delta}}} \right\vert_{SV}= \delta+1$. Hence, due to the remark above, we have that for all ordinals $\gamma < \beta$ there exists an ordinal $\gamma<\delta< \beta$ such that $\left\vert  \frac{1}{4} y_{T_{i_{\delta}}}[I_{i_{\delta}}] \right\vert_{SV}= \delta+1$. Since $y'_T= p' + \frac{1}{4} \sum_{n=0}^{\infty} y'_{T_n}[I_n]$, $p'$ is continuous on all $\operatorname{cl}(I_n)$ for all $n \in \mathbb{N}_0$, and all the intervals are disjoint, this means that for all ordinals $\gamma < \beta$ there exists an ordinal $\gamma<\delta< \beta$ such that $\operatorname{cl}(I_{i_{\delta}}) \cap E_\delta \neq \emptyset$ and $\operatorname{cl}(I_{i_{\delta}}) \cap E_\alpha = \emptyset$. Note that this also implies $\operatorname{cl}(I_{i}) \cap E_\alpha = \emptyset$ for all $i \in \mathbb{N}_0$.


What has been described above implies exactly that for the Cantor set we have $C= \bigcap_{\gamma <\beta} E_\gamma$ which by definition means that $C=E_\beta$. Let us now analyze the behavior of  $y'_T \restriction_{E_\beta}$. We observe that $y'_T \restriction_{E_\beta}$ is continuous on all $I_n$ for all $n \in \mathbb{N}_0$. Indeed, we have just shown that for all $i \in \mathbb{N}_0$ we have $I_i \cap E_\alpha = \emptyset$ and so that implies continuity of $y'_T \restriction_{E_\beta}$ on such intervals. We are now left to consider the behavior of $y'_T \restriction_{E_\beta}$ on the Cantor set $C$. It is easy to see that $y'_T \restriction_{C}=p'\restriction_{C}$. Indeed each other term in the sum equals zero on the Cantor set by the definition of each function $y'_{T_n}[I_n]$ for all $n \in \mathbb{N}_0$. But since $p'\restriction_{C}$ is discontinuous on the set $\{ 0,1 \}$, this means that $y'_T \restriction_{E_\beta}$ is discontinuous on the set $\{ 0,1 \}$ and so $E_\alpha =\{ 0,1 \}$. Note that by definition we know that $y'_T \restriction_{E_\alpha}(0)=y'_T \restriction_{E_\alpha}(1)=0$. It follows $E_{\alpha+1} = \emptyset$ and so $\left\vert y_T \right\vert_{SV}= \alpha +1$. 

\end{enumerate}






\item  $\sup_n \left\vert T_n \right\vert_{ls}= \alpha$ \\

In this second case instead, we infer that there exists a finite set $K \subset \mathbb{N}_0$ such that all the subtrees $T_i$ for all $i \in K$ satisfy $\left\vert T_i \right\vert_{ls}= \alpha$ while for all the other subtrees $T_i$ for all $i \in \mathbb{N}_0 \setminus K$ we have $\left\vert T_i \right\vert_{ls} \leq \beta$. By induction hyphotesis we know that for all the subtrees $T_i$ for all $i \in \mathbb{N}_0$ such that $\left\vert T_i \right\vert_{ls}= \gamma$ with $\gamma< \alpha$ we have $\left\vert y_{T_i} \right\vert_{SV}= \gamma +1$. Therefore, this means that for all the subtrees $T_i$ for all $i \in \mathbb{N}_0 \setminus K$ we have $\left\vert y_{T_i} \right\vert_{SV} \leq \alpha$. Hence, due to the remark above, this implies that for all $i \in \mathbb{N}_0 \setminus K$ we have $\left\vert \frac{1}{4} y_{T_i}[I_i] \right\vert_{SV} \leq \alpha$. Since function $p'$ is discontinuous only on 0 and 1, this means by definition of $y_T$ that, for all intervals $I_i$ for all $i \in \mathbb{N}_0 \setminus K$ we have $\operatorname{cl}(I_i) \cap E_\alpha = \emptyset$. 

Therefore, we only consider the intervals $I_i$ for all $i \in K$, for which we know that $\left\vert T_i \right\vert_{ls}= \alpha$. But for all $i \in K$, if $\left\vert T_i \right\vert_{ls}= \alpha$, that means that once again for each tree $T_i$ we only have the two possibilities mentioned in the list above. Let us call $Q \subseteq K$ the set of trees $T_i$ with $\left\vert T_i \right\vert_{ls}= \alpha$ and that belong to the second case. We have just proved above that for all $i \in K \setminus Q$ the trees $T_i$ that belong to the first case of the list satisfies $\left\vert y_{T_i} \right\vert_{SV}= \alpha +1$ and so $\left\vert \frac{1}{4} y_{T_i}[I_i] \right\vert_{SV}= \alpha+1$. Moreover, since $p'$ is continuous on $I_i$, by looking carefully at the analysis of the previous case in the point above it is clear that each of these intervals contributes to two points in $E_{\alpha}$, i.e.\ the rational extremes of the interval $I_i$ which belong to $C$. Let us now analyze what happens to the intervals $I_i$ for all $i \in Q$. For each of these subtrees, we can repeat the same reasoning, but this process has to end because the original tree $T$ is well-founded. It is easy to see that supposing the process to end while always encountering only subtrees belonging to the second case leads to contradiction. Indeed, a tree formed by a single root has necessarily a limsup ranking equal to $1 < \alpha$. In other words, this implies that eventually every tree belonging to the second case has to reduce to a collection of subtrees belonging to the first case, which means that each interval $I_i$ for all $i \in Q$ contains a countable set $L_i \subset I_i$ of rational points belonging to $E_\alpha$. 

In conclusion, we have shown that $E_\alpha= \bigcup_{i \in Q} L_i \bigcup_{i \in K \setminus Q} \{ a_i, b_i \}$. Moreover, by the analysis of the first case done above it immediately follows that $y'_T \restriction_{E_\alpha}=p'\restriction_{E_\alpha}$. But since we know that for all $i \in Q$ we have $L_i \subset I_i$ we also know that $\{ 0,1 \} \cap E_\alpha = \emptyset$. By definition of $p'$ this impies that $y'_T \restriction_{E_\alpha}$ is continuous, and therefore $E_{\alpha+1} = \emptyset$. Hence, also in this case we have $\left\vert y_T \right\vert_{SV}= \alpha +1$. 

\end{enumerate}

\end{proof}

\subsection{Properties of the solvable ranking}

The goal of this section is now to describe the set descriptive complexity of the set of solvable functions $SV$ as a subset of $C([0,1])$. We show that similarly with the case of the set of differentiable functions $\mathcal{D}$, the set $SV$ is a $\boldsymbol{\Pi}_1^1$ subset of $C([0,1])$ which is not Borel. Moreover, we also show that the solvable ranking as introduced in Definition \ref{def:solvablerank} is a $\boldsymbol{\Pi}_1^1$ norm on $SV$. Both these conclusions are direct consequences of the results mentioned above; we express them with the two following propositions. 

\begin{proposition}
The set of solvable functions $SV$ is a $\boldsymbol{\Pi}_1^1$ subset of $C([0,1])$ which is not Borel. 
\end{proposition}

\begin{proof}
To show the statement of the proposition we make use of a consequence of Theorem \ref{thm:diffranking} and Proposition \ref{prop:unboundedsolv} i.e.\ the fact that the $\boldsymbol{\Pi}_1^1$ norm generated by Kechris-Woodin ranking is uncountable on the set of solvable functions $SV$. Indeed, since $SV \subset \mathcal{D}$ and $\mathcal{D}$ is a $\boldsymbol{\Pi}_1^1$ subset of $C([0,1])$, in order to show that $SV$ is $\boldsymbol{\Pi}_1^1$ but not Borel it is enough to construct a norm $\varphi: SV \rightarrow$ ORD such that:
\begin{enumerate}
\item there is a $\boldsymbol{\Sigma}_1^1$ relation $\prec$ over $C([0,1])^2$ such that:
 
$$x, y \in SV \Longrightarrow[\varphi(x)<\varphi(y) \Longleftrightarrow x \prec y]$$

\item the range of $\varphi$ is uncountable.
\end{enumerate}

Therefore, if we choose as such norm the Kechris-Woodin ranking restricted to $SV$, i.e.\ $\varphi = \left\vert \cdot \right\vert_{KW} \upharpoonright {SV}$ then the first item is immediately satisfied since such ranking is a $\boldsymbol{\Pi}_1^1$ norm on $\mathcal{D}$. Instead, the second item is satisfied as an immediate consequence of Theorem \ref{thm:diffranking} combined with Proposition \ref{prop:unboundedsolv}. 
\end{proof}

\begin{proposition}
The solvable ranking is a $\boldsymbol{\Pi}_1^1$ norm on $SV$. 
\end{proposition}

\begin{proof}
Given two Polish spaces $X, Y$, two $\boldsymbol{\Pi}_1^1$ sets $P \subseteq X, Q \subseteq Y$, a Borel function $f: X \rightarrow Y$ such that $x \in P \Leftrightarrow f(x) \in Q$ and given a norm $\varphi: Q \rightarrow$ ORD we can define a norm $\psi: P \rightarrow$ ORD by $\psi(x)=\varphi(f(x))$. It is well known \cite{kechris1986ranks} that if $\psi$ is a $\boldsymbol{\Pi}_1^1$ norm on $P$ then $\varphi$ is a $\boldsymbol{\Pi}_1^1$ norm on $Q$. We can now take as the two $\boldsymbol{\Pi}_1^1$ sets above the sets $WF$ and $SV$ respectively and make use of the Borel map $y_T: WF \to SV$ constructed in Definition \ref{def:treefunction}. Proposition \ref{prop:unboundedsolv} implies that, if we choose as norm on $SV$ the solvable ranking as defined in Definition \ref{def:solvablerank}, then the norm $\psi: WF \rightarrow$ ORD defined as $\psi(T)=\left\vert T \right\vert_{ls}+1$ satisfies $\psi(T) = \left\vert y_T \right\vert_{SV}$. Then, since the norm $\psi$ defined this way is a known $\boldsymbol{\Pi}_1^1$ norm, it follows that the solvable ranking is a $\boldsymbol{\Pi}_1^1$ norm on $SV$. 
\end{proof}

\section{Conclusion}

\label{sec:conclusions}

In this work we have expanded the study of solvable functions as introduced in our recent papers \cite{StacsBournezGozzi2024}, \cite{bournez2024solvable}. Specifically, we provided a ranking for the class of solvable functions that naturally translates to solvable dynamical systems in the form of IVPs ruled by solvable ODEs. With Proposition \ref{prop:unboundedsolv} we showed that for each countable ordinal, there exists at least one solvable function with ranking equal to such ordinal, therefore formally proving one of the intuitive claims of \cite{StacsBournezGozzi2024}, \cite{bournez2024solvable}. Indeed, starting from the functions discussed in Proposition \ref{prop:unboundedsolv} and built with the help of Definition \ref{def:treefunction}, it is trivial to describe solvable IVPs in higher dimensions with the same ranking by using techniques similar to the one employed in Example \ref{ex:firstex}. Hence, this proves that the hierarchy of solvable IVPs introduced by our ranking is indeed fully populated, as already envisioned (albeit only in an intuitive manner) in \cite{StacsBournezGozzi2024}, \cite{bournez2024solvable}. 

Moreover, by limiting the scope of our ranking to the case of one-dimensional solvable functions, we illustrated several properties and comparisons with other rankings for differentiable functions presented in the literature. More precisely, we demonstrated that the set of one-dimensional solvable functions is a coanalytic subset of the space of continuous functions, and the solvable ranking introduces a coanalytic norm over such a subset. Moreover, we showed with Theorem \ref{thm:diffranking} that the solvable ranking dominates the Kechris-Woodin ranking, consequently dominating also the Denjoy ranking and the Zalcwasser ranking due to known results from \cite{ki1997denjoy}. Finally, an insightful implication of the results of this paper is the following fact: solving analytically IVPs involving discontinuous ODEs requires the totality of the countable ordinals, similarly to what has been already proved by \cite{dougherty1991complexity} for the case of Denjoy integration. 

To conclude, we outline three relevant possible directions for future work: 

\begin{enumerate}
\item The analysis of the set descriptive complexity of solvable functions performed in this paper could be further developed in another direction by considering its lightface complexity (i.e.\ constructivity). This could mean producing a finer-grained investigation on the complexity of each level of the hierarchy inspired by the existing treatments for differentiability and Denjoy integration described in \cite{westrick2014lightface} and \cite{Wes20} respectively. 
\item One interesting question to address could be proving that the set of solvable functions $SV$ is a $\boldsymbol{\Pi}_1^1$ complete subset of $C([0,1])$ as it is the case for $\mathcal{D}$. One possible way to prove such a statement could be to provide a one-to-one reduction of $SV$ to the set of well-founded trees $WF$ such as the one described for $\mathcal{D}$ in \cite{kechris1986ranks}. Practically, this means designing, in addition to the Borel encoding of well-founded trees into solvable functions produced by Definition \ref{def:treefunction}, a similar encoding for the other direction, i.e.\ from solvable functions to well-founded trees. Nonetheless, one difficulty to overcome to achieve such encoding is that every solvable function is also a differentiable function, and so the tree description of a solvable function needs to include a tree description of differentiability (such as the one in \cite{kechris1986ranks}) as well as the extra condition of solvability. 
\item Finally, having proved that the hierarchy of solvable IVPs introduced by our ranking is fully populated, an interesting next step could be to make use of the IVP described in Example \ref{ex:secondex} as a building block to produce, for each hyperarithmetical real, a simulation using a solvable IVP that reaches such real. This could be done by rescaling and composing the dynamics used in Example \ref{ex:secondex} in the spirit of continuous-time ODEs simulations presented in \cite{TAMC06}. A similar achievement could represent the equivalent of what has already been obtained in \cite{dougherty1991complexity} for Denjoy integration, producing a description for the lightface descriptive complexity of solvable function alternative to the one mentioned in point 1 of this list. Such investigation could therefore yield a continuous-time, dynamical characterization of ordinal computing using systems of discontinuous ODEs. 
\end{enumerate}

\bibliographystyle{plainurl}
\bibliography{perso,bournez}
\end{document}